\documentclass[11pt,dvips,twoside,letterpaper]{article}

\usepackage{pslatex}
\usepackage{fancyhdr}
\usepackage{graphicx}
\usepackage{geometry}

\def\figurename{Figure} % Replace the colon that normally appears after the Figure number by a period.
\makeatletter
\renewcommand{\fnum@figure}[1]{\figurename~\thefigure.}
\makeatother

\def\tablename{Table} % Replace the colon that normally appears after the Figure number by a period.
\makeatletter
\renewcommand{\fnum@table}[1]{\tablename~\thetable.}
\makeatother

\usepackage{amsmath}
\usepackage{amssymb}
\usepackage{amsfonts}
\usepackage{amsthm,amscd}

\newtheorem{theorem}{Theorem}[section]
\newtheorem{lemma}[theorem]{Lemma}
\newtheorem{corollary}[theorem]{Corollary}
\newtheorem{proposition}[theorem]{Proposition}
\theoremstyle{definition}
\newtheorem{definition}[theorem]{Definition}
\newtheorem{example}[theorem]{Example}

\theoremstyle{remark}

\numberwithin{equation}{section}

\begin{document}
%\vskip 0.4in
\title{\bfseries{Extremum conditions for functionals  involving higher derivatives of several variable vector valued functions}}
\author{\bfseries\scshape Mahouton Norbert Hounkonnou\thanks{E-mail address: norbert.hounkonnou@cipma.uac.bj}\\
International Chair in Mathematical Physics
and Applications\\ (ICMPA-UNESCO Chair)\\University of Abomey-Calavi\\072 B.P. 50  Cotonou, Republic of Benin\\ \\
\bfseries\scshape Pascal Dkengne Sielenou\thanks{E-mail address: sielenou$\_$alain@yahoo.fr}\\
International Chair in Mathematical Physics
and Applications\\ (ICMPA-UNESCO Chair)\\University of Abomey-Calavi\\072 B.P. 50  Cotonou, Republic of Benin\\ \\
%\bfseries\scshape Author 3\thanks{E-mail address: author3@math.univ.yzedu}\\
%Department of Mathematics\\YZ University\\Washington, DC 20059,
%USA
 \\ \\}

\date{}
\maketitle \thispagestyle{empty} \setcounter{page}{1}

% ------- [First Page Running Head] - place it immediately after title! ------
\thispagestyle{fancy} \fancyhead{}
%\fancyhead[L]{{\bf {\Huge C}ommunications in {\Huge M}athematical {\Huge A}nalysis}\\
%{\it Special Volume in Honor of Prof. Peter Lax}\\
%Volume .., pp. {\thepage--\pageref{lastpage-01}} (2011)\\
%ISSN \ 1938-9787
%} % put \label{lastpage-xx} on the last page!
%\fancyhead[R]{\small \sf www.commun-math-anal.org\\}\fancyfoot{}
\renewcommand{\headrulewidth}{0pt}

\begin{abstract} \noindent
This paper addresses both  necessary and relevant sufficient extremum conditions for  a variational problem
defined by a smooth Lagrangian, involving higher derivatives of several variable vector valued functions.  A general formulation of
first order necessary extremum conditions for variational problems  with (or without) constraints is given.
Global  Legendre second order  necessary  extremum conditions are provided as well as new general explicit formula for second
order sufficient extremum condition which does not require the notion of conjugate points as in the Jacobi sufficient condition.
\end{abstract}


\vspace{.08in} \noindent \textbf{Keywords}: Variational problems, critical points, extremum conditions.

\section{Introduction}

The calculus of variations encompasses a very broad range of mathematical applications.
The methods of variational analysis can be applied to an enormous variety of physical systems,
whose equilibrium configurations inevitably minimize or maximize a suitable functional
which typically represents the potential energy of the system.
The critical functions are characterized as solutions to a system of partial differential equations,
known as the Euler-Lagrange equations associated with the variational principle.
Each solution to the  problem specified by the Euler-Lagrange
equations subject to appropriate boundary conditions is thus a candidate for extrema of the functional defining the variational
problem.  In many applications, the Euler-Lagrange boundary value problem suffices to
single out the physically relevant solutions, and one does not need to press onto the considerably
more difficult second variation.

In general, the solutions to the Euler-Lagrange boundary value problem are critical
functions for the functional defining the variational problem, and hence include all (smooth) local and global extrema.
The determination of which solutions are genuine minima or maxima requires  further analysis
of the positivity properties of the second variation.
Indeed, as stated  in \cite{r00}, a complete analysis of the positive definiteness of the second
variation of multi-dimensional variational problems is quite complicated, and still awaits
a completely satisfactory resolution!
This  is thus a reason for which  second order conditions of extrema are  customary  established only for
functional whose Lagrangian involves dependent variables together with at most their first order derivatives \cite{r00,r01,r02,r03,r04}.
The aim  of this paper is to give some satisfactory expressions of the second order extremum conditions for a
functional whose Lagrangian also depends on the higher order derivatives of the dependent variables.

%%%%%%%%%%%%%%%%%%%%%%%%%%%%%%%%%%%%%%%%%%%%%%%%%%%%%%%%%%%%%%%%%%%%%%%%%%%%%%%%%%%%%%%%%%%%%%%
\section{Brief review of known results}
%----------------------------------------------
\subsection{Holonomic constraints}
We  consider functional of the form
\begin{equation}\label{eqiii1}
 \mathcal{F}(u)=\int_{a}^{b} F\left(x,u(x),u'(x)\right) dx,
 \end{equation}
where $u\in \mathcal{C}^2\left(\overline{I},\mathbb{R}^N\right),$ and $I=]a,b[.$
We  demand that $u$ satisfies a holonomic constraint
\begin{equation}\label{eqiii2}
 g(x,u(x))=0, \quad a\leq x\leq b.
 \end{equation}
\begin{theorem}[\cite{r04}]
Suppose that $F\in \mathcal{C}^2\left(\overline{I}\times\Omega\right),$ where $\Omega$ is an open
set in $\mathbb{R}^{2N}.$ Suppose that $g\in \mathcal{C}^2\left(\overline{I}\times W\right),$ where
$W\subset \mathbb{R}^N$ and that $\nabla_u g(x,u)\neq 0$ on the set where $g(x,u(x))=0.$ Suppose that
$u\in \mathcal{C}^2\left(\overline{I}, W\right)$    is a local extremum  for $\mathcal{F},$ subject to
the holonomic constraint in (\ref{eqiii2}).  Then there is a function $\lambda\in \mathcal{C}\left(\overline{I}\right)$
such that $u$ is an extremum of the functional
\begin{equation}\label{eqiii3}
 \mathcal{G}(u)=\int_{a}^{b}\left[ F\left(x,u(x),u'(x)\right) +\lambda(x) g(x,u(x))\right]dx.
 \end{equation}
\end{theorem}
\textsc{Remark.}
 The Lagrangian of the functional $\mathcal{G}$ in (\ref{eqiii3}) is
 $$ G(x,u,u')=F(x,u,u')+\lambda(x) g(x,u)  $$
 and the Euler-Lagrange equations are
 $$ F_{u^j}+\lambda g_{u^j}-\frac{d}{dx}F_{u'^j}=0,\quad j=1,2,  \cdots  ,N.  $$
%-------------------------------------------------------------------------------
\subsection{Nonholonomic constraints} $ $\\
\begin{theorem}[\cite{r04}]
Suppose that $F,$ and $g^j$ for $j=1,2,  \cdots  ,m$ belong to $\mathcal{C}^3\left(\overline{I}\times\Omega,\mathbb{R}     \right),$
where $\Omega\in \mathbb{R}^{2N}$ and that $u\in \mathcal{C}^2([a,b],\mathbb{R}^N)$ is a local extremum of the functional
\begin{equation}\label{eqiii4}
 \mathcal{F}(u)=\int_{a}^{b} F\left(x,u(x),u'(x)\right) dx,
 \end{equation}
subject to the nonholonomic constraints
$$  g^j(x,u(x),u'(x))=0,\quad j=1,2,  \cdots  ,m. $$
Suppose that the constraints together with $u$ satisfy the following properties.
\begin{itemize}
\item[(1)] The matrix
$$ D_ug(x,u,u')=\left(  \frac{\partial g^j (x,u,u')}{\partial u'^k }  \right)   $$
has rank $m$ for $a \leq x \leq b;$
\item[(2)] The only solutions to the system of differential equations
$$  \sum_{j=1}^m  \left[ \left( g^j_{u^k}-\frac{d }{dx}g^j_{u'^k}   \right)\mu_j-g^j_{u'^k}\frac{d \mu_j}{dx}      \right]=0,\quad k=1,2, \cdots  ,N   $$
are $\mu_1(x)=\mu_2(x)=  \cdots  =\mu_m(x)=0.$
\end{itemize}
Then there exist functions $\lambda_1,\lambda_2,   \cdots, \lambda_m$ defined on $[a,b]$ such that $u$ is an extremum for the
functional with Lagrangian
$$ G(x,u,u')=F(x,u,u')+\sum_{j=1}^m\lambda_j(x)g^j(x,u,u'). $$
\end{theorem}
%-------------------------------------------------------------
\subsection{The Legendre condition} $ $\\
\begin{theorem}[\cite{r04}]
Suppose that $u$ is a local, weak minimum for the functional
$$
 \mathcal{F}(u)=\int_{a}^{b} F\left(x,u(x),u'(x)\right) dx.
 $$
Then
\begin{equation}\label{eqiii5}
 \sum_{j,k=1}^N F_{u'^ju'^k}\left(x,u(x),u'(x)\right)\xi^j\xi^k \geq 0 ,\quad \forall\, a\leq x \leq b,\,\, \forall\,\xi\in \mathbb{R}^N.
 \end{equation}
\end{theorem}
The inequality in (\ref{eqiii5}) is called the Legendre condition. As the theorem says, it is a necessary condition for
$u$  to be a weak minimum. The Legendre condition says that the matrix
$$F_{u'u'}=\left(F_{u'^ju'^k}\right)$$
 must be positive semi-definite at every point along a minimum.
 %-----------------------------------------------------
 \subsection{The Jacobi conditions}
Consider the functional
\begin{equation}\label{eqiii6}
 \mathcal{F}(u)=\int_{a}^{b} F\left(x,u(x),u'(x)\right) dx,
 \end{equation}
where $u=\left(u^1,u^2,\cdots,u^n\right).$
Introduce the matrices
$$F_{uu}=\left( F_{u^iu^k}  \right),\quad F_{uu'}=\left( F_{u^iu'^k}  \right),\quad F_{u'u'}=\left( F_{u'^iu'^k}  \right),$$
$$P= \frac{1}{2}F_{u'u'},\quad Q=\frac{1}{2}\left( F_{uu}-\frac{d }{dx}F_{uu'}  \right).$$
\begin{definition}
Let
\begin{eqnarray}
h^1&=&\left( h_{11},h_{12}, \cdots  ,h_{1n}   \right)\nonumber\\
h^2&=&\left( h_{21},h_{22}, \cdots  ,h_{2n}   \right)\nonumber\\
\vdots & & \cdots  \label{eqiii7}\\
h^n&=&\left( h_{n1},h_{n2}, \cdots  ,h_{nn}   \right)\nonumber
\end{eqnarray}
be set of $n$ solutions of the linear equations called the Jacobi system
\begin{equation}\label{eqiii8}
-\frac{d }{dx}\left( P\,h' \right)+Q\,h=0
\end{equation}
associated with the functional (\ref{eqiii6}), where the $i$-th solution satisfies the initial conditions
$$  h_{ik}(a)=0,\,\,h'_{ii}(a)=1,\,\, h'_{ik}(a)=0,\quad k \neq i,\,\,i,k=1,2,\cdots,n.   $$
Then the point $\widetilde{a},$ $(\widetilde{a}\neq a),$ is said to be conjugate to the point $a$ if the
determinant
$$ \left| \begin{array}{cccc}
h_{11}(x)& h_{12}(x)&\cdots&h_{1n}(x)\\
h_{21}(x)& h_{22}(x)&\cdots&h_{2n}(x)\\
\vdots &       &     \cdots &\\
h_{n1}(x)& h_{n2}(x)&\cdots&h_{nn}(x)
   \end{array} \right|  $$
vanishes for $x=\widetilde{a}.$
\end{definition}
\begin{theorem}[Jacobi necessary condition \cite{r01}]
 If the extremum $u$ corresponds to a minimum of the functional
(\ref{eqiii6}), and if the matrix $P(x,u(x),u'(x))$ is positive definite along this extremum, then the open
interval $]a,b[$ contains no points conjugate to $a.$
\end{theorem}
\begin{theorem}[Jacobi sufficient condition \cite{r01}]
Suppose that for some curve $\gamma$ with equation $u=u(x),$ the functional (\ref{eqiii6}) satisfies
the following conditions:
\begin{itemize}
\item[(1)] The curve $\gamma$  is an extremum, i.e., satisfies the system of Euler equations
$$F_{u^i}-\frac{d }{dx}F_{u'^i}=0,\quad i=1,2,   ,n;  $$
\item[(2)] Along $\gamma$ the matrix
$$ P(x)=\frac{1}{2}F_{u'u'}(x,u(x),u'(x))  $$
is positive definite;
\item[(3)] The interval $[a,b]$ contains no points conjugate to the point $a.$
\end{itemize}
Then the functional (\ref{eqiii6}) has a weak minimum for the curve $\gamma.$
\end{theorem}
$ $\\

In this work, we give an answer to the following question:
What do  the results of the four above theorems become when the vector-valued function
$u=\left(u^1,  \cdots   ,u^m\right)$ depends on several variables $x=\left(x^1, \cdots ,x^n\right)$
and/or the Lagrangian of the used functional includes higher order derivatives of $u?$
To our best knowledge of the literature, in this general situation, there is not  explicit method available to determine
if a known extremum is a minimum, a maximum, or a saddle point.
To fill this gap  and provide a suitable answer to our main question, we   establish a
regular connection between   the second variation of a functional and an operational  square matrix.
Therefore, by the well known result of the matrix theory, explicit formula for the
necessary and sufficient extremum conditions can be derived without making use of the notion
of conjugate points as in the Jacobi theorems. Furthermore, the matrices $F_{uu},$ $F_{uu'}$ and $F_{u'u'}$ used in
the above Legendre and Jacobi conditions are deduced as submatrices of a
general matrix  associated with the second variation.

%%%%%%%%%%%%%%%%%%%%%%%%%%%%%%%%%%%%%%%%%%%%%%%%%%%%%%%%%%%%%%%%%%%%%%%%%%%%%%%%%%%%%%%%%%%%%%%
\section{Notations for partial derivatives of functions}
$ $\\
 Consider $X,$ an $n$-dimensional independent variable space, and $U,$
 an $m$-dimensional dependent variable space.
 Let $x=\left(x^1,\cdots,x^n\right)\in X$ and $u=\left(u^1,\cdots,u^m\right)\in U.$
  We define  the space $U^{(s)},$ $s\in \mathbb{N}$ as:
  \begin{equation}
   U^{(s)}:=\left\{u^{(s)}\,:\,\,u^{(s)}=\bigotimes_{j=1}^{m}\left(\bigotimes_{k=0}^{s}u^j_{(k)}  \right)\right\},
  \end{equation}
 where $u^j_{(k)}$ is the
 \begin{equation}\label{sa1}
 p_k= \left(\begin{array}{c}{n+k-1}\\{k}  \end{array}\right)\emph{\emph{-tuple}}
 \end{equation}
  of all distinct $k$-order partial derivatives of $u^j.$
The   $u^j_{(k)}$ vector components are recursively obtained  as follows:
 \begin{itemize}
 \item[i)] $u^j_{(0)}=u^j$ and $u^j_{(1)}=\left(u^j_{x^1},u^j_{x^2},\cdots,u^j_{x^n}\right).$
 \item[ii)] Assume that  $u^j_{(k)}$ is known.
 \begin{itemize}
\item Form the tuples $\widehat{u}^j_{(k+1)}(l):$
 $$\widehat{u}^j_{(k+1)}(l)=\left(\frac{\partial}{\partial x^1}u^j_{(k)}[l], \frac{\partial}{\partial x^2}u^j_{(k)}[l],\cdots,\frac{\partial}{\partial x^n}u^j_{(k)}[l]   \right),\quad l=1,2,\cdots,p_k;   $$
 where $u^j_{(k)}[l]$ is the $l$-th component of the vector $u^j_{(k)}.$
 \item Construct, by iteration, the tuples $\widetilde{u}^j_{(k+1)}(l):$
%following the scheme
    $\widetilde{u}^j_{(k+1)}(1)= \widehat{u}^j_{(k+1)}(1)$ and for $l=2,3,\cdots,p_k,$
    the vector $\widetilde{u}^j_{(k+1)}(l)$ is nothing but the tuple $\widehat{u}^j_{(k+1)}(l)$ in which all
    components already present in $\widetilde{u}^j_{(k+1)}(i),$ $i=1,2, \cdots, l-1,$ are excluded.
\item Finally, form the vector
$$u^j_{(k+1)}=\left( \widetilde{u}^j_{(k+1)}(1),\widetilde{u}^j_{(k+1)}(2),\cdots,\widetilde{u}^j_{(k+1)}(p_k)   \right).  $$
\end{itemize}
 \end{itemize}
As a matter of clarity, let us immediately illustrate this construction by the following.
\begin{example}
\begin{itemize}
 \item For
   $n=2$,  $x=\left(x^1,x^2\right)$ and we have:
   $$u^j_{(1)}=\left(u^j_{x^1},u^j_{x^2}\right),$$
\begin{eqnarray}
 \widehat{u}^j_{(2)}(1)&=&\left(\frac{\partial}{\partial x^1}u^j_{(1)}[1], \frac{\partial}{\partial x^2}u^j_{(1)}[1]   \right)\,=\,\left(u^j_{2x^1},u^j_{x^1x^2}\right),\nonumber
\end{eqnarray}
\begin{eqnarray}
 \widehat{u}^j_{(2)}(2)&=&\left(\frac{\partial}{\partial x^1}u^j_{(1)}[2], \frac{\partial}{\partial x^2}u^j_{(1)}[2]   \right)\,=\,\left(u^j_{x^2x^1},u^j_{2x^2}\right),\nonumber
\end{eqnarray}
$$ \widetilde{u}^j_{(2)}(1)= \widehat{u}^j_{(2)}(1)=\left(u^j_{2x^1},u^j_{x^1x^2}\right),\quad \widetilde{u}^j_{(2)}(2)=\left(\check{u}^j_{x^2x^1},u^j_{2x^2}\right) =\left(u^j_{2x^2}\right),$$
$$u^j_{(2)}=\left(\widetilde{u}^j_{(2)}(1), \widetilde{u}^j_{(2)}(2)\right)
 =\left(u^j_{2x^1},u^j_{x^1x^2},u^j_{2x^2}\right).$$
  \item  For $n=3$, $x=(x^1,x^2,x^3)$ and the same scheme leads to
 $$u^j_{(2)}=\left(u^j_{2x^1},u^j_{x^1x^2},u^j_{x^1x^3},u^j_{2x^2},u^j_{x^2x^3},u^j_{2x^3}  \right),$$
 $$u^j_{(3)}=\left(u^j_{3x^1},u^j_{2x^1x^2},u^j_{2x^1x^3},u^j_{x^12x^2},u^j_{x^1x^2x^3},u^j_{x^12x^3},u^j_{3x^2},u^j_{2x^2x^3},u^j_{x^22x^3},u^j_{3x^3}  \right),$$
for $k=2$ and $k=3$, respectively.
\end{itemize}
\end{example}

  An element   $u^{(s)},$ in the space $U^{(s)},$ is the
  \begin{equation}\label{sa2}
  q_s= m(1+p_1+p_2+\cdots+p_s)=m \left(\begin{array}{c}{n+s}\\{s}  \end{array}\right)\emph{\emph{-tuple}}
  \end{equation}
  defined by
\begin{equation}\label{sa3}
u^{(s)}=\left(u^1_{(0)},u^1_{(1)},\cdots,u^1_{(s)}, u^2_{(0)},u^2_{(1)},\cdots,u^2_{(s)},\cdots,u^m_{(0)},u^m_{(1)},\cdots,u^m_{(s)}  \right).
\end{equation}
We denote by $X\times U^{(s)},$ the total space whose coordinates are
denoted by $(x,u^{(s)}),$ encompassing    the independent variables $x$ and the
dependent variables with their derivatives up to order $s,$ globally denoted by  $u^{(s)}.$ \\

In the sequel, a $q_s$-uple $u^{(s)}$ is referred to (\ref{sa3}), whereas the integers $p_k$ and $q_s$ are  defined by (\ref{sa1}) and (\ref{sa2}), respectively.

%%%%%%%%%%%%%%%%%%%%%%%%%%%%%%%%%%%%%%%%%%%%%%%%%%%%%%%%%%%%%%%%%%%%%%%%%%%%%%%%%%%%%%%%%%%%%%%%%%
%%%%%%%%%%%%%%%%%%%%%%%%%%%%%%%%%%%%%%%%%%%%%%%%%%%%%%%%%%%%%%%%%%%%%%%%%%%%%%%%%%%%%%%%%%%%%%%%%%
\section{First variation and  necessary conditions for local extrema}
This section contains two parts. First, we briefly recall useful definitions and properties used in the sequel. Then, we analyze
the variational problem with constraints, and  give a general formulation of the first order necessary extremum condition
 which is rigorously proved.

%%%%%%%%%%%%%%%%%%%%%%%%%%%%%%%%%%%%%%%%%%%%%%%%%%%%%%%%%%%%%%%%%%%%%%%%%%%%%%%%%%%%%%%%%%%%%%%%%%%
\subsection{Variational problem without constraints: definitions and main results}

Consider a functional of the form
\begin{equation}\label{se1}
\mathcal{F}(u)=\int_{\Lambda}L\left(x,u^{(s)}(x)\right)dx
\end{equation}
where $\Lambda$ is a  connected subset of $X.$ Let $\Omega$ be an open subset of $U^{(s)}.$ We assume that
the function $L,$ usually called the Lagrangian of the functional $\mathcal{F},$ is defined on the open subset
$\Lambda\times \Omega$ of $X\times U^{(s)}$ and is continuous in all its $n+q_s$ variables so that the
variational integral (\ref{se1}) exists. The problem consists in  finding conditions that the function $u$ must satisfy
in order to be a minimum or maximum of the functional $\mathcal{F},$  requiring  that
   $L\in \mathcal{C}^{s+1}(\Lambda\times \Omega,\mathbb{R}).$
For the integral in (\ref{se1})  be defined, it is necessary that the function $u\in \mathcal{C}_b^s(\Lambda,U),$
where
$$  \mathcal{C}_b^s(\Lambda,U)=\left\{ \psi\in   \mathcal{C}^s(\Lambda,U) \,:\, \sum_{j=1}^m \sum_{k=0}^s \sum_{l=1}^{p_k}\sup_{x\in\Lambda}\left| \psi^j_{(k)}[l](x) \right|<+\infty \right\}   .$$
In addition, $L(x,u^{(s)}(x))$ must be defined for all  $x\in \Lambda.$ This means that $u^{(s)}(x)\in \Omega$ for all $x\in \Lambda.$
Such a function $u$ is said to be admissible for the functional $\mathcal{F}.$
\begin{definition}
A function $u$ which is admissible for the functional $\mathcal{F}$ is a global minimum for $\mathcal{F},$ if
$  \mathcal{F}(u)\leq \mathcal{F}(v)$     for every  admissible function $v.$
\end{definition}
\begin{definition}
A function $u$ which is admissible for the functional $\mathcal{F}$ is a global maximum for $\mathcal{F},$ if
$  \mathcal{F}(v)\leq \mathcal{F}(u)$     for every  admissible function $v.$
\end{definition}
A function which is either a global minimum or a global maximum is called a global extremum.
To come up with the definition of local extrema for a functional, we need to have  a measure
of distance between two functions.
\begin{definition}
 Let $\phi\in \mathcal{C}_b^s(\Lambda,U).$ We define the $0$-norm of $\phi$ by
 $$ \|  \phi \|_0 =\sum_{j=1}^m \sup_{x\in \Lambda}\left| \phi^j(x) \right|$$
 and the $s$-norm of $\phi$ by
 $$\|  \phi \|_s=\sum_{j=1}^m \sum_{k=0}^s \sum_{l=1}^{p_k}\sup_{x\in\Lambda}\left| \phi^j_{(k)}[l](x) \right|   .$$
\end{definition}
Clearly, for $s>0$ the numbers $\|   \phi-\psi \|_0$ and $\|   \phi-\psi \|_s$  provide quite different measures of
the distance between $\phi$ and $\psi.$ These measures lead to two different definitions of local minima.
\begin{definition}
A function $u$ which is admissible for the functional $\mathcal{F}$ is a weak local minimum for $\mathcal{F}$
if there is an $\epsilon >0$ such that $\mathcal{F}(u)\leq \mathcal{F}(v)$ for all admissible functions $v$
satisfying $\| v-u  \|_s< \epsilon.$
$u$ is a strict weak local minimum if $\mathcal{F}(u)< \mathcal{F}(v)$ for all such $v$ with $v\neq u.$
\end{definition}
\begin{definition}
A function $u$ which is admissible for the functional $\mathcal{F}$ is a strong local minimum for $\mathcal{F}$
if there is an $\epsilon >0$ such that $\mathcal{F}(u)\leq \mathcal{F}(v)$ for all admissible functions $v$
satisfying $\| v-u  \|_0< \epsilon.$
$u$ is a strict strong local minimum if $\mathcal{F}(u)< \mathcal{F}(v)$ for all such $v$ with $v\neq u.$
\end{definition}
\begin{definition}
A function $u$ which is admissible for the functional $\mathcal{F}$ is a weak local maximum for $\mathcal{F}$
if there is an $\epsilon >0$ such that $\mathcal{F}(u)\geq \mathcal{F}(v)$ for all admissible functions $v$
satisfying $\| v-u  \|_s< \epsilon.$
$u$ is a strict weak local maximum if $\mathcal{F}(u)> \mathcal{F}(v)$ for all such $v$ with $v\neq u.$
\end{definition}
\begin{definition}
A function $u$ which is admissible for the functional $\mathcal{F}$ is a strong local maximum for $\mathcal{F}$
if there is an $\epsilon >0$ such that $\mathcal{F}(u)\geq \mathcal{F}(v)$ for all admissible functions $v$
satisfying $\| v-u  \|_0< \epsilon.$
$u$ is a strict strong local maximum if $\mathcal{F}(u)> \mathcal{F}(v)$ for all such $v$ with $v\neq u.$
\end{definition}
A function which is either a weak local minimum or a weak local maximum is called a weak local extremum.
A function which is either a strong local minimum or a strong local maximum is called a strong local extremum.\\

Without loss of generality, we can assume that $\Lambda=\prod_{i=1}^n \left]a^i,b^i  \right[$ with  $a^i\leq b^i.$
\begin{definition}
 A function $\psi\in \mathcal{C}(\Lambda,U)$ is said to have  compact support in $\Lambda$ if there is $\epsilon>0$
 such that  $\psi(x)=0$ for all $x=\left(x^1,  \cdots ,x^n\right)$ with $x^i\in \left]a^i, a^i+\epsilon  \right[$ or $ x^i\in \left]b^i-\epsilon,b^i   \right[ $ for some
 $i\in \{ 1,2, \cdots  ,n  \}.$
 The set of all functions which are infinitely differentiable and have compact support in $\Lambda$ is denoted by $\mathcal{C}^{\infty}_0(\Lambda,U).$
\end{definition}
\begin{lemma}\label{lem1}
Let $f\in \mathcal{C}(\Lambda,\mathbb{R}).$  If $\int_{\Lambda}f(x)\psi(x)dx=0$ for all $\psi\in \mathcal{C}^{\infty}_0(\Lambda,\mathbb{R}),$
then $f(x)=0$ for all $x\in \Lambda.$
\end{lemma}
Given an admissible function $u\in \mathcal{C}^{s}(\Lambda,U)$ and any $\phi\in \mathcal{C}^{\infty}_0(\Lambda,U),$ there is an $\epsilon_0>0$
such that the function $v=u+t\,\phi$ is admissible for all $|t|<\epsilon_0.$
Therefore, the function
\begin{equation}\label{se2}
\Phi(t)=\mathcal{F}(u+t\phi)=\int_{\Lambda}L\left(x,u^{(s)}(x)+t\,\phi^{(s)}(x)\right)dx
\end{equation}
is a well defined function of $t$ for $|t|<\epsilon_0.$
Throughout this paper, $\epsilon_0$ stands for  such a number.

Assume now that $u\in \mathcal{C}^{s}(\Lambda,U)$ is a local extremum of $\mathcal{F}.$ We may as well assume that $u$
is a local minimum. We have $\Phi(t)=\mathcal{F}(u+t\phi)\geq \mathcal{F}(u)=\Phi(0)$ for  $|t|<\epsilon_0,$ i.e. $0$ is a local
minimum for $\Phi.$ Suppose that  $L\in \mathcal{C}^1(\Lambda\times \Omega,\mathbb{R})$ implying that $\Phi$ is also continuously differentiable
 and we must have
\begin{equation}\label{se3}
 \Phi'(0)=0.
 \end{equation}
We can calculate $\Phi'$ by differentiating (\ref{se2}) with respect to $t$ under the integral sign. Doing so
and using the chain rule we get
\begin{eqnarray}
 \Phi'(t)&=&\frac{d}{d t}\mathcal{F}(u+t\, \phi)\nonumber\\
 &=& \frac{d}{d t}\int_{\Lambda}L\left(x,u^{(s)}(x)+t\,\phi^{(s)}(x)\right)dx  \nonumber\\
 &=& \int_{\Lambda}\frac{d}{d t} L\left(x,u^{(s)}(x)+t\,\phi^{(s)}(x)\right)dx  \nonumber\\
             &=& \int_{\Lambda}\sum_{j=1}^{m}\sum_{k=0}^{s}\sum_{h=1}^{p_k}\,\phi^j_{(k)}[h](x)\,\frac{\partial L\left(x,u^{(s)}(x)+t\, \phi^{(s)}(x)\right)}{\partial u^j_{(k)}[h]}dx.\label{se10}
\end{eqnarray}
In particular at $t=0$ we get
\begin{equation}\label{se4}
\Phi'(0)= \int_{\Lambda}\sum_{j=1}^{m}\sum_{k=0}^{s}\sum_{h=1}^{p_k}\,\phi^j_{(k)}[h](x)\,\frac{\partial L\left(x,u^{(s)}(x)\right)}{\partial u^j_{(k)}[h]}dx.
\end{equation}
\begin{definition}\label{df}
 The first variation of $\mathcal{F}$ in a neighborhood of $u$ in the direction $\phi$ is defined by
 \begin{equation}\label{se5}
 \delta \mathcal{F}(u+t\, \phi,\phi)=\Phi'(t) .
 \end{equation}
 In particular, the first variation of $\mathcal{F}$ at $u$ in the direction $\phi$ is expressed by
 \begin{equation}\label{se6}
 \delta \mathcal{F}(u,\phi)=\Phi'(0) .
 \end{equation}
\end{definition}
Notice that the first variation at $u$ is defined in \textrm{Definition} \ref{df} whether $u$ is a local extremum or not.
However, if $u$ is a local extremum of $\mathcal{F},$ then by (\ref{se3}) and (\ref{se6}), $\delta \mathcal{F}(u,\phi)=0.$
We have proved the following first order necessary condition on a local extremum of $\mathcal{F}.$
\begin{proposition}
 Suppose that $L\in \mathcal{C}^1(\Lambda\times \Omega,\mathbb{R}),$ and  that $u\in \mathcal{C}_b^{s}(\Lambda,U)$ is a local
 extremum for the functional $\mathcal{F}(u)=\int_{\Lambda}L\left(x,u^{(s)}(x)\right)dx.$ Then
 \begin{equation}\label{se7}
 \delta \mathcal{F}(u,\phi)=0
 \end{equation}
for all $\phi\in \mathcal{C}^{\infty}_0(\Lambda,U).$
\end{proposition}
The condition in (\ref{se7}) is called the weak form of the Euler-Lagrange equations. A function $u$ which satisfies
(\ref{se7}) is called the weak extremum of $\mathcal{F}.$

Now  assume that the Lagrangian $L\in \mathcal{C}^{s+1}(\Lambda\times \Omega,\mathbb{R}),$ and $u\in \mathcal{C}_b^{2s}(\Lambda,U).$
Using the divergence theorem to  successively  integrate by parts  (\ref{se4}) until  all  derivative actions  on $\phi^j$
are now moved into
$\frac{\partial L\left(x,u^{(s)}(x)\right)}{\partial u^j_{(k)}[h]},$ and taking into account that
$\phi^j\in \mathcal{C}^{\infty}_0(\Lambda,\mathbb{R}),$ we get
\begin{equation}\label{se8}
\delta \mathcal{F}(u,\phi)=\int_{\Lambda}\sum_{j=1}^{m}  \left( \sum_{k=0}^{s}(-1)^k\sum_{h=1}^{p_k} \left( \frac{\partial L\left(x,u^{(s)}(x)\right)}{\partial u^j_{(k)}[h]}   \right)_{(k)}[h]  \right) \phi^j(x)dx.
\end{equation}
If $u$ is a weak local extremum, then (\ref{se8}) is equal to $0$ for all $\phi\in \mathcal{C}^{\infty}_0(\Lambda,U).$
In particular if we take $\phi=\psi\,e^l,$ where $\psi\in \mathcal{C}^{\infty}_0(\Lambda,\mathbb{R})$ and $e^l$ is
the $l$-th vector of the canonical basis of $\mathbb{R}^m,$ then we get
$$ 0=\delta \mathcal{F}(u,\psi\,e^l)=\int_{\Lambda}  \left( \sum_{k=0}^{s}(-1)^k\sum_{h=1}^{p_k} \left( \frac{\partial L\left(x,u^{(s)}(x)\right)}{\partial u^l_{(k)}[h]}   \right)_{(k)}[h]  \right) \psi(x)dx    $$
for all $\psi\in \mathcal{C}^{\infty}_0(\Lambda,\mathbb{R}).$
By \textrm{Lemma} \ref{lem1}, we see that
$$  \sum_{k=0}^{s}(-1)^k\sum_{h=1}^{p_k} \left( \frac{\partial L\left(x,u^{(s)}(x)\right)}{\partial u^l_{(k)}[h]}   \right)_{(k)}[h]  =0   $$
for all $x\in \Lambda$ and $l=1,2,   \cdots,m.$
Thus, we have proved the following theorem.
\begin{theorem}
Suppose that $L\in \mathcal{C}^{s+1}(\Lambda\times \Omega,\mathbb{R}),$ and $u\in \mathcal{C}_b^{2s}(\Lambda,U)$
is a local extremum for the functional $\mathcal{F}(u)=\int_{\Lambda}L(x,u^{(s)}(x))dx.$ Then
\begin{equation}\label{se9}
\sum_{k=0}^{s}(-1)^k\sum_{h=1}^{p_k} \left( \frac{\partial L\left(x,u^{(s)}(x)\right)}{\partial u^j_{(k)}[h]}   \right)_{(k)}[h]  =0
\end{equation}
for all $x\in \Lambda$ and $j=1,2,   \cdots,m.$
\end{theorem}
The equations (\ref{se9}) are called the Euler-Lagrange equations. A solution to the Euler-Lagrange equations
is called an extremum for the functional $\mathcal{F}.$

%%%%%%%%%%%%%%%%%%%%%%%%%%%%%%%%%%%%%%%%%%%%%%%%%%%%%%%%%%%%%%%%%%%%%%%%%%%%%%%%%%%%%%%%%%%%%%%%%%
\subsection{Variational problem with constraints: main results}

We want to find extrema for the functional
\begin{equation}\label{ssse1}
\mathcal{F}(u)=\int_{\Lambda}L\left(x,u^{(s_1)}(x)\right)dx
\end{equation}
subject to constraints of the form
\begin{equation}\label{ssse2}
F_j\left(x,u^{(s_2)}(x)\right)=0\qquad j=1,2, \cdots  ,m'
\end{equation}
for all $x\in \Lambda.$ Let $\Omega_i$ be open subsets of $U^{(s_i)},$ $i=1,\,2$ such that $L$ is defined on $\Lambda\times\Omega_1$ and
$F_j$ is defined on $\Lambda\times\Omega_2.$ Constraints of type  (\ref{ssse2})
are called holonomic constraints if $s_2=0,$ and nonholonomic constraints if $s_2\geq 1.$
In this subsection, we examine these types of constrained variational problems.

For $m=m',$ i.e. the number of equations in the system formed by the constraints is equal
to the number of unknowns, we  exploit the fact that such a system appears for the Euler-Lagrange equations
 of  some variational problems  \cite{r06,r07} to  prove our next result.
\begin{theorem}
Suppose that $L\in\mathcal{C}^{s_1+1}(\Lambda\times \Omega_1,\mathbb{R}),$   $F_j\in\mathcal{C}^{s_2+1}(\Lambda\times \Omega_2,\mathbb{R})$ and that the function $u\in \mathcal{C}_b^{2s}(\Lambda,U),$ $s=\max(s_1,s_2),$
verifies the constraints (\ref{ssse2}) and is a local extremum for the functional $\mathcal{F}$ defined by (\ref{ssse1}).
If a function $\lambda(x)= \left( \lambda^1(x),   \cdots ,\lambda^m(x)  \right)$ defined on $\Lambda$ is solution of the system
\begin{equation}\label{ssse3}
\sum_{k=0}^{s_2}(-1)^k\sum_{h=1}^{p_k} \left( \frac{\partial \left[\sum_{l=1}^m\,\lambda^l(x)\,F_l\left(x,u^{(s_2)}(x)\right)\right]}{\partial u^j_{(k)}[h]}   \right)_{(k)}[h]  =0\quad j=1,2,   \cdots,m
\end{equation}
 then $u$ is a local extremum for the functional whose Lagrangian is
 \begin{equation}\label{ssse4}
 G\left(x,u^{(s)}(x)\right)=L\left(x,u^{(s_1)}(x)\right)+\sum_{l=1}^m\,\lambda^l(x)\,F_l\left(x,u^{(s_2)}(x)\right).
\end{equation}
\end{theorem}
\begin{proof}
Consider the variational problem whose Lagrangian is defined by
\begin{equation}\label{ssse5}
 G'\left(x,u^{(s)}(x),v(x)\right)=L\left(x,u^{(s_1)}(x)\right)+\sum_{l=1}^m\,v^l(x)\,F_l\left(x,u^{(s_2)}(x)\right),
\end{equation}
where $v(x)=\left(v^1(x), \cdots ,v^m(x)   \right)$ is viewed as dependent variable. The Euler-Lagrange equations of this variational
problem are
\begin{equation}\label{ssse6}
 P_j\equiv \sum_{k=0}^{s}(-1)^k\sum_{h=1}^{p_k} \left( \frac{\partial G'\left(x,u^{(s)}(x),v(x)\right)}{\partial u^j_{(k)}[h]}   \right)_{(k)}[h]  =0;
\end{equation}
\begin{equation}\label{ssse7}
 Q_j\equiv \sum_{k=0}^{s}(-1)^k\sum_{h=1}^{p_k} \left( \frac{\partial G'\left(x,u^{(s)}(x),v(x)\right)}{\partial v^j_{(k)}[h]}   \right)_{(k)}[h]  =0,
\end{equation}
$j=1,2,   \cdots,m.$ Taking into account (\ref{ssse5}), the expressions of $P_j$ and $Q_j$ give
$$ P_j=P_{j,1}+P_{j,2};\quad  Q_j=Q_{j,1}+Q_{j,2} $$
where
$$ P_{j,1}= \sum_{k=0}^{s_1}(-1)^k\sum_{h=1}^{p_k} \left( \frac{\partial L\left(x,u^{(s_1)}(x)\right)}{\partial u^j_{(k)}[h]}   \right)_{(k)}[h] ;$$
$$ P_{j,2}= \sum_{k=0}^{s_2}(-1)^k\sum_{h=1}^{p_k} \left( \frac{\partial \left[\sum_{l=1}^m\,v^l(x)\,F_l\left(x,u^{(s_2)}(x)\right)\right]}{\partial u^j_{(k)}[h]}   \right)_{(k)}[h]; $$
$$ Q_{j,1}= \sum_{k=0}^{s_1}(-1)^k\sum_{h=1}^{p_k} \left( \frac{\partial L\left(x,u^{(s_1)}(x)\right)}{\partial v^j_{(k)}[h]}   \right)_{(k)}[h] ;$$
$$ Q_{j,2}= \sum_{k=0}^{s_2}(-1)^k\sum_{h=1}^{p_k} \left( \frac{\partial \left[\sum_{l=1}^m\,v^l(x)\,F_l\left(x,u^{(s_2)}(x)\right)\right]}{\partial v^j_{(k)}[h]}   \right)_{(k)}[h].  $$

The $P_{j,1}$ are  expressions defining the Euler-Lagrange equations of the variational problem (\ref{ssse1}). Thus, $P_{j,1}=0$ since
 $u$ is a local extremum for the functional $\mathcal{F}.$
 According to the relations (\ref{ssse3}), the expressions $P_{j,2}$ vanish when $v(x)=\lambda(x).$
 The expressions $Q_{j,1}$ vanish since the Lagrangian $L$ does  depend neither on $v$ nor  on its derivatives.

 For $j=1,2,   \cdots,m,$ $Q_{j,2}=F_j\left(x,u^{(s_2)}(x)\right)$ and therefore vanish since the function $u$ satisfies
 the constraints (\ref{ssse2}).

 Finally, the Euler-Lagrange equations (\ref{ssse6})-(\ref{ssse7}) are automatically verified if and only if $v(x)=\lambda(x).$
 This proves that $u$ is also a local extremum for the functional whose Lagrangian is
 $G'\left(x,u^{(s)}(x),\lambda(x)\right)=G\left(x,u^{(s)}(x)\right).$
\end{proof}

For $m'<m,$ we redefine the problem in the following manner:
Find the extrema for the functional
\begin{equation}\label{sssd1}
\mathcal{F}(u,\widetilde{u})=\int_{\Lambda}L\left(x,u^{(s_1)}(x),\widetilde{u}^{(s_1)}(x)\right)dx
\end{equation}
subject to the constraints
\begin{equation}\label{sssd2}
F_j\left(x,u^{(s_2)}(x),\widetilde{u}^{(s_2)}(x)\right)=0\qquad j=1,2, \cdots  ,m
\end{equation}
for all $x\in \Lambda,$ where $\widetilde{u}(x)=\left(\widetilde{u}^1(x), \cdots  ,\widetilde{u}^{\widetilde{m}}(x)\right)\in \widetilde{U},$
$\widetilde{U}$ being an $\widetilde{m}$-dimensional space.
Let $\Omega_i$ be open subsets of $U^{(s_i)}$ and $\widetilde{\Omega}_i$ be open subsets of $\widetilde{U}^{(s_i)},$ $i=1,\,2$ such that $L$ is defined on $\Lambda\times\Omega_1\times \widetilde{\Omega}_1$ and
$F_j$ is defined on $\Lambda\times\Omega_2\times \widetilde{\Omega}_2.$
Here, the number of equations in the system formed by the constraints is lower than
 the number of unknowns, i.e. the constraints form an under-determined system.
 Such a system   appears for the Euler-Lagrange equations
 of some variational problems  \cite{r08}.  We then prove the following result.
\begin{theorem}
Suppose that $L\in\mathcal{C}^{s_1+1}(\Lambda \times \Omega_1 \times \widetilde{\Omega}_1,\mathbb{R}),$ $F_j\in\mathcal{C}^{s_2+1}(\Lambda \times \Omega_2 \times \widetilde{\Omega}_2,\mathbb{R})$ and that the function $(u,\widetilde{u})\in \mathcal{C}_b^{2s}(\Lambda,U) \times \mathcal{C}_b^{2s}(\Lambda,\widetilde{U}),$  $s=\max(s_1,s_2),$
verifies the constraints (\ref{sssd2}) and is a local extremum for the functional $\mathcal{F}$ defined by (\ref{sssd1}).
If a function $\left(\lambda(x),\widetilde{\lambda}(x)\right)$ defined on $\Lambda$ with\\
    $\lambda(x)= \left( \lambda^1(x),   \cdots ,\lambda^m(x)  \right)$ and $\widetilde{\lambda}(x)= \left( \widetilde{\lambda}^1(x),   \cdots ,\widetilde{\lambda}^{\widetilde{m}}(x)  \right),$
  is solution to the system
  \begin{equation}\label{sssd3}
\sum_{k=0}^{s}(-1)^k\sum_{h=1}^{p_k} \left( \frac{\partial \left[\sum_{l=1}^m\left(\lambda^l +\sum_{\widetilde{l}=1}^{\widetilde{m}}\,\widetilde{\lambda}^{\widetilde{l}} \right)F_l\left(x,u^{(s_2)},\widetilde{u}^{(s_2)}\right)\right]}{\partial u^j_{(k)}[h]}   \right)_{(k)}[h]  =0;
\end{equation}
\begin{equation}\label{sssd4}
\sum_{k=0}^{s}(-1)^k\sum_{h=1}^{p_k} \left( \frac{\partial \left[\sum_{l=1}^m\left(\lambda^l +\sum_{\widetilde{l}=1}^{\widetilde{m}}\,\widetilde{\lambda}^{\widetilde{l}} \right)F_l\left(x,u^{(s_2)},\widetilde{u}^{(s_2)}\right)\right]}{\partial \widetilde{u}^{\widetilde{j}}_{(k)}[h]}   \right)_{(k)}[h]  =0,
\end{equation}
$j=1,2,   \cdots,m,$ $\widetilde{j}=1,2,   \cdots,\widetilde{m},$
 then $(u, \widetilde{u})$ is a local extremum for the functional whose Lagrangian is
 \begin{eqnarray}
 G\left(x,u^{(s)}(x),\widetilde{u}^{(s)}(x)\right)&=& \sum_{l=1}^m\left(\lambda^l(x) +\sum_{\widetilde{l}=1}^{\widetilde{m}}\,\widetilde{\lambda}^{\widetilde{l}}(x) \right)F_l\left(x,u^{(s_2)}(x),\widetilde{u}^{(s_2)}(x)\right)\nonumber\\
 &+& L\left(x,u^{(s_1)}(x),\widetilde{u}^{(s_1)}(x)\right).\label{sssd5}
\end{eqnarray}
\end{theorem}
\begin{proof}
Consider the variational problem whose Lagrangian is defined by
\begin{eqnarray}
 G'\left(x,u^{(s)},\widetilde{u}^{(s)},v,\widetilde{v}\right)&=& \sum_{l=1}^m\left(v^l(x) +\sum_{\widetilde{l}=1}^{\widetilde{m}}\,\widetilde{v}^{\widetilde{l}}(x) \right)F_l\left(x,u^{(s_2)}(x),\widetilde{u}^{(s_2)}(x)\right)\nonumber\\
 &+& L\left(x,u^{(s_1)}(x),\widetilde{u}^{(s_1)}(x)\right),\label{sssd6}
\end{eqnarray}
where $v(x)=\left(v^1(x), \cdots ,v^m(x)   \right),$  and $\widetilde{v}(x)=\left(\widetilde{v}^1(x), \cdots ,\widetilde{v}^{\widetilde{m}}(x)   \right)$ are viewed as dependent variables. The Euler-Lagrange equations of this variational
problem are
\begin{equation}\label{sssd7}
 P_j\equiv \sum_{k=0}^{s}(-1)^k\sum_{h=1}^{p_k} \left( \frac{\partial G'\left(x,u^{(s)}(x),\widetilde{u}^{(s)}(x),v(x),\widetilde{v}(x)\right)}{\partial u^j_{(k)}[h]}   \right)_{(k)}[h]  =0;
\end{equation}
\begin{equation}\label{sssd8}
 Q_{\widetilde{j}}\equiv \sum_{k=0}^{s}(-1)^k\sum_{h=1}^{p_k} \left( \frac{\partial G'\left(x,u^{(s)}(x),\widetilde{u}^{(s)}(x),v(x),\widetilde{v}(x)\right)}{\partial \widetilde{u}^{\widetilde{j}}_{(k)}[h]}   \right)_{(k)}[h]  =0,
\end{equation}
\begin{equation}\label{sssd9}
 R_j\equiv \sum_{k=0}^{s}(-1)^k\sum_{h=1}^{p_k} \left( \frac{\partial G'\left(x,u^{(s)}(x),\widetilde{u}^{(s)}(x),v(x),\widetilde{v}(x)\right)}{\partial v^j_{(k)}[h]}   \right)_{(k)}[h]  =0;
\end{equation}
\begin{equation}\label{sssd10}
 S_{\widetilde{j}}\equiv \sum_{k=0}^{s}(-1)^k\sum_{h=1}^{p_k} \left( \frac{\partial G'\left(x,u^{(s)}(x),\widetilde{u}^{(s)}(x),v(x),\widetilde{v}(x)\right)}{\partial \widetilde{v}^{\widetilde{j}}_{(k)}[h]}   \right)_{(k)}[h]  =0,
\end{equation}
$j=1,2,   \cdots,m,$ $\widetilde{j}=1,2,   \cdots,\widetilde{m}.$ Taking into account (\ref{sssd6}), the expressions of $P_j,$  $Q_{\widetilde{j}},$
$R_j$ and $S_{\widetilde{j}}$ are given by
$$ P_j=P_{j,1}+P_{j,2};\quad  Q_{\widetilde{j}}=Q_{\widetilde{j},1}+Q_{\widetilde{j},2} ;$$
$$ R_j=R_{j,1}+R_{j,2};\quad  S_{\widetilde{j}}=S_{\widetilde{j},1}+S_{\widetilde{j},2} $$
where
$$ P_{j,1}= \sum_{k=0}^{s_1}(-1)^k\sum_{h=1}^{p_k} \left( \frac{\partial L\left(x,u^{(s_1)}(x),\widetilde{u}^{(s_1)}(x)\right)}{\partial u^j_{(k)}[h]}   \right)_{(k)}[h] ;$$

$$  P_{j,2}= \sum_{k=0}^{s_2}(-1)^k\sum_{h=1}^{p_k}
 \left( \frac{\partial \left[\sum_{l=1}^m\left(v^l +\sum_{\widetilde{l}=1}^{\widetilde{m}}\,\widetilde{v}^{\widetilde{l}} \right)F_l\left(x,u^{(s_2)},\widetilde{u}^{(s_2)}\right)\right]}{\partial u^j_{(k)}[h]}   \right)_{(k)}[h]  ;  $$
 $$ Q_{\widetilde{j},1}= \sum_{k=0}^{s_1}(-1)^k\sum_{h=1}^{p_k} \left( \frac{\partial L\left(x,u^{(s_1)}(x),\widetilde{u}^{(s_1)}(x)\right)}{\partial \widetilde{u}^{\widetilde{j}}_{(k)}[h]}   \right)_{(k)}[h] ;$$
 $$Q_{\widetilde{j},2}= \sum_{k=0}^{s_2}(-1)^k\sum_{h=1}^{p_k}
 \left( \frac{\partial \left[\sum_{l=1}^m\left(v^l +\sum_{\widetilde{l}=1}^{\widetilde{m}}\,\widetilde{v}^{\widetilde{l}} \right)F_l\left(x,u^{(s_2)},\widetilde{u}^{(s_2)}\right)\right]}{\partial \widetilde{u}^{\widetilde{j}}_{(k)}[h]}   \right)_{(k)}[h]  ;   $$
 $$ R_{j,1}= \sum_{k=0}^{s_1}(-1)^k\sum_{h=1}^{p_k} \left( \frac{\partial L\left(x,u^{(s_1)},\widetilde{u}^{(s_1)}\right)}{\partial v^j_{(k)}[h]}   \right)_{(k)}[h] ;$$
 $$  R_{j,2}= \sum_{k=0}^{s_2}(-1)^k\sum_{h=1}^{p_k}
 \left( \frac{\partial \left[\sum_{l=1}^m\left(v^l +\sum_{\widetilde{l}=1}^{\widetilde{m}}\,\widetilde{v}^{\widetilde{l}} \right)F_l\left(x,u^{(s_2)},\widetilde{u}^{(s_2)}\right)\right]}{\partial v^j_{(k)}[h]}   \right)_{(k)}[h]  ; $$
 $$ S_{\widetilde{j},1}= \sum_{k=0}^{s_1}(-1)^k\sum_{h=1}^{p_k} \left( \frac{\partial L\left(x,u^{(s_1)}(x),\widetilde{u}^{(s_1)}(x)\right)}{\partial \widetilde{v}^{\widetilde{j}}_{(k)}[h]}   \right)_{(k)}[h] ;$$
 $$  S_{\widetilde{j},2}= \sum_{k=0}^{s_2}(-1)^k\sum_{h=1}^{p_k}
 \left( \frac{\partial \left[\sum_{l=1}^m\left(v^l +\sum_{\widetilde{l}=1}^{\widetilde{m}}\,\widetilde{v}^{\widetilde{l}} \right)F_l\left(x,u^{(s_2)},\widetilde{u}^{(s_2)}\right)\right]}{\partial \widetilde{v}^{\widetilde{j}}_{(k)}[h]}   \right)_{(k)}[h]  . $$
The $P_{j,1}$ and $Q_{\widetilde{j},1}$ are  nothing but the Euler-Lagrange equations of the variational problem (\ref{sssd1}). Hence, $P_{j,1}=0$ and $Q_{\widetilde{j},1}=0$ since
 $(u,\widetilde{u})$ is a local extremum for the functional $\mathcal{F}.$

 According to the relations (\ref{sssd3}) and (\ref{sssd4}), the expressions $P_{j,2}$ and $Q_{\widetilde{j},2}$ vanish when $\left(v(x),\widetilde{v}(x)\right)=\left(\lambda(x), \widetilde{\lambda}(x)\right).$

 The expressions $R_{j,1}$ and $S_{\widetilde{j},1}$  vanish since the Lagrangian $L$ does  depend neither on $v$ and $\widetilde{v}$ nor on their derivatives.

 For $j=1,2,   \cdots,m,$ and $\widetilde{j}=1,2,   \cdots,\widetilde{m},$ we have $R_{j,2}=F_j\left(x,u^{(s_2)}(x),\widetilde{u}^{(s_2)}(x)\right)$ and $S_{\widetilde{j},2}=\sum_{l=1}^mF_l\left(x,u^{(s_2)}(x),\widetilde{u}^{(s_2)}(x)\right)$ which therefore vanish since the function $u$ satisfies
 the constraints (\ref{sssd2}).

 Finally, the Euler-Lagrange equations (\ref{sssd7})-(\ref{sssd10}) are automatically verified if and only if $\left(v(x),\widetilde{v}(x)\right)=\left(\lambda(x), \widetilde{\lambda}(x)\right).$
 This proves that $u$ is also a local extremum for the functional whose Lagrangian is
 $G'\left(x,u^{(s)}(x),\widetilde{u}^{(s)}(x),\lambda(x),\widetilde{\lambda}(x)\right)=G\left(x,u^{(s)}(x),\widetilde{u}^{(s)}(x)\right).$
\end{proof}
%%%%%%%%%%%%%%%%%%%%%%%%%%%%%%%%%%%%%%%%%%%%%%%%%%%%%%%%%%%%%%%%%%%%%%%%%%%%%%%%%%%%%%%%%%%%%%%%%%
%%%%%%%%%%%%%%%%%%%%%%%%%%%%%%%%%%%%%%%%%%%%%%%%%%%%%%%%%%%%%%%%%%%%%%%%%%%%%%%%%%%%%%%%%%%%%%%%%%
\section{Second variation and    conditions for local extrema: main results}
This section contains relevant  results which are new to our best knowledge of the literature.
We  investigate  the second variation of a functional as well as  the necessary and sufficient conditions that
a function should satisfy to be either a minimum or a maximum.\\

 Consider a variational problem of the form (\ref{se1}) with the  Lagrangian $L\in \mathcal{C}^2(\Lambda\times \Omega,\mathbb{R}).$
 Define an $m\times m$ block matrix $A$ of second order partial derivatives of $L$ by:
\begin{equation}\label{sse0}
A=\left[ A^{jj'} \right]_{1\leq j,\,j'\leq m}
\end{equation}
with $A^{jj'}$ being again an $s\times s$ block matrix defined by
$$ A^{jj'}=\left[ A^{jj'}_{kk'} \right]_{0\leq k,\,k'\leq s},$$
where $A^{jj'}_{kk'}$ is a $p_k\times p_{k'}$ matrix defined by
 \begin{equation}\label{sse00}
A^{jj'}_{kk'}=\left[ \frac{\partial^2L}{\partial u^j_{(k)}[h]\partial u^{j'}_{(k')}[h']} \right]_{_{1\leq h'\leq p_{k'}\,}^{ 1\leq h\leq p_k  }}.
\end{equation}
Note that the  matrix   $A$ is obviously symmetric by construction.

\begin{example}
 Let us construct the matrix $A^{jj'}$ for particular values of the integers $n$ and $s.$
 If $s=1$, then
 $$  A^{jj'}=\left[\begin{array}{cc} A^{jj'}_{00}& A^{jj'}_{01}\\
 A^{jj'}_{10}&A^{jj'}_{11}\end{array}\right].  $$
 In this case, we have
  for $n=1,$  $x=x^1:$
$$\begin{array}{lll}
A^{jj'}_{00}=\frac{\partial^2L}{\partial u^j   \partial u^{j'}},
&A^{jj'}_{01}=\frac{\partial^2L}{\partial u^j   \partial u^{j'}_{x}},\\
A^{jj'}_{10}=\frac{\partial^2L}{\partial u^j_{x}   \partial u^{j'}},
&A^{jj'}_{11}=\frac{\partial^2L}{\partial u^j_{x}   \partial u^{j'}_{x}}, \end{array}$$
thus
$$  A^{jj'}=\left[\begin{array}{cc} \frac{\partial^2L}{\partial u^j   \partial u^{j'}}& \frac{\partial^2L}{\partial u^j   \partial u^{j'}_{x}}\\
 \frac{\partial^2L}{\partial u^j_{x}   \partial u^{j'}}&\frac{\partial^2L}{\partial u^j_{x}   \partial u^{j'}_{x}}\end{array}\right];  $$
 For $n=2,$  $x=(x^1,x^2):$
 $$\begin{array}{ll} A^{jj'}_{00}=\frac{\partial^2L}{\partial u^j   \partial u^{j'}},&
A^{jj'}_{01}=\left(\begin{array}{ll}\frac{\partial^2L}{\partial u^j   \partial u^{j'}_{x^1}}& \frac{\partial^2L}{\partial u^j   \partial u^{j'}_{x^2}}   \end{array}\right)\end{array},$$
$$\begin{array}{ll}
A^{jj'}_{10}=\left(\begin{array}{l}\frac{\partial^2L}{\partial u^j_{x^1}   \partial u^{j'}}\\
 \frac{\partial^2L}{\partial u^j_{x^1}   \partial u^{j'}}\end{array}\right)  , &A^{jj'}_{11}=\left(\begin{array}{ll}\frac{\partial^2L}{\partial u^j_{x^1}   \partial u^{j'}_{x^1}}&
\frac{\partial^2L}{\partial u^j_{x^1}   \partial u^{j'}_{x^2}}\\
\frac{\partial^2L}{\partial u^j_{x^2}   \partial u^{j'}_{x^1}}&\frac{\partial^2L}{\partial u^j_{x^2}   \partial u^{j'}_{x^2}}         \end{array}\right)\end{array},$$
 thus
 $$  A^{jj'}=\left[\begin{array}{cc} \frac{\partial^2L}{\partial u^j   \partial u^{j'}}& \begin{array}{ll}\frac{\partial^2L}{\partial u^j   \partial u^{j'}_{x^1}}& \frac{\partial^2L}{\partial u^j   \partial u^{j'}_{x^2}}   \end{array}\\
 \begin{array}{l}\frac{\partial^2L}{\partial u^j_{x^1}   \partial u^{j'}}\\
 \frac{\partial^2L}{\partial u^j_{x^1}   \partial u^{j'}}\end{array}&\begin{array}{ll}\frac{\partial^2L}{\partial u^j_{x^1}   \partial u^{j'}_{x^1}}&
\frac{\partial^2L}{\partial u^j_{x^1}   \partial u^{j'}_{x^2}}\\
\frac{\partial^2L}{\partial u^j_{x^2}   \partial u^{j'}_{x^1}}&\frac{\partial^2L}{\partial u^j_{x^2}   \partial u^{j'}_{x^2}}         \end{array}\end{array}\right].  $$

 If $s=2$, then
 $$ \quad A^{jj'}=\left[\begin{array}{ccc} A^{jj'}_{00}& A^{jj'}_{01}&A^{jj'}_{02}\\
 A^{jj'}_{10}&A^{jj'}_{11}&A^{jj'}_{12}\\A^{jj'}_{20}&A^{jj'}_{21}&A^{jj'}_{22}\end{array}\right].  $$
 In this case, we have
  for $n=1,$  $x=x^1:$
$$\begin{array}{lll}
A^{jj'}_{00}=\frac{\partial^2L}{\partial u^j   \partial u^{j'}},
&A^{jj'}_{01}=\frac{\partial^2L}{\partial u^j   \partial u^{j'}_{x}},
&A^{jj'}_{02}=\frac{\partial^2L}{\partial u^j   \partial u^{j'}_{2x}},\\
A^{jj'}_{10}=\frac{\partial^2L}{\partial u^j_{x}   \partial u^{j'}},
&A^{jj'}_{11}=\frac{\partial^2L}{\partial u^j_{x}   \partial u^{j'}_{x}},
&A^{jj'}_{12}=\frac{\partial^2L}{\partial u^j_{x}   \partial u^{j'}_{2x}},\\
A^{jj'}_{20}=\frac{\partial^2L}{\partial u^j_{2x}   \partial u^{j'}},
&A^{jj'}_{21}=\frac{\partial^2L}{\partial u^j_{2x}   \partial u^{j'}_{x}},
&A^{jj'}_{22}=\frac{\partial^2L}{\partial u^j_{2x}   \partial u^{j'}_{2x}}, \end{array}$$
thus
$$ A^{jj'}=\left[\begin{array}{ccc} \frac{\partial^2L}{\partial u^j   \partial u^{j'}}& \frac{\partial^2L}{\partial u^j
  \partial u^{j'}_{x}}&\frac{\partial^2L}{\partial u^j   \partial u^{j'}_{2x}}\\
 \frac{\partial^2L}{\partial u^j_{x}   \partial u^{j'}}&\frac{\partial^2L}{\partial u^j_{x}   \partial u^{j'}_{x}}&\frac{\partial^2L}{\partial u^j_{x}   \partial u^{j'}_{2x}}\\
 \frac{\partial^2L}{\partial u^j_{2x}   \partial u^{j'}}&\frac{\partial^2L}{\partial u^j_{2x}   \partial u^{j'}_{x}}&\frac{\partial^2L}{\partial u^j_{2x}   \partial u^{j'}_{2x}}\end{array}\right];  $$
For $n=2,$  $x=(x^1,x^2):$
$$\begin{array}{ll} A^{jj'}_{00}=\frac{\partial^2L}{\partial u^j   \partial u^{j'}},&
A^{jj'}_{01}=\left(\begin{array}{ll}\frac{\partial^2L}{\partial u^j   \partial u^{j'}_{x^1}}& \frac{\partial^2L}{\partial u^j   \partial u^{j'}_{x^2}}   \end{array}\right)\end{array},$$
$$\begin{array}{l}A^{jj'}_{02}=\left(\begin{array}{lll}\frac{\partial^2L}{\partial u^j   \partial u^{j'}_{2x^1}}&
\frac{\partial^2L}{\partial u^j   \partial u^{j'}_{x^1x^2}}&\frac{\partial^2L}{\partial u^j   \partial u^{j'}_{2x^2}}   \end{array}\right)\end{array}    ,$$
$$\begin{array}{ll}
A^{jj'}_{10}=\left(\begin{array}{l}\frac{\partial^2L}{\partial u^j_{x^1}   \partial u^{j'}}\\
 \frac{\partial^2L}{\partial u^j_{x^1}   \partial u^{j'}}\end{array}\right)  , &A^{jj'}_{11}=\left(\begin{array}{ll}\frac{\partial^2L}{\partial u^j_{x^1}   \partial u^{j'}_{x^1}}&
\frac{\partial^2L}{\partial u^j_{x^1}   \partial u^{j'}_{x^2}}\\
\frac{\partial^2L}{\partial u^j_{x^2}   \partial u^{j'}_{x^1}}&\frac{\partial^2L}{\partial u^j_{x^2}   \partial u^{j'}_{x^2}}         \end{array}\right)\end{array},$$
$$A^{jj'}_{12}=\left(\begin{array}{lll}
\frac{\partial^2L}{\partial u^j_{x^1}   \partial u^{j'}_{2x^1}}&\frac{\partial^2L}{\partial u^j_{x^1}   \partial u^{j'}_{x^1x^2}}&\frac{\partial^2L}{\partial u^j_{x^1}   \partial u^{j'}_{2x^2}}\\
\frac{\partial^2L}{\partial u^j_{x^2}   \partial u^{j'}_{2x^1}}&\frac{\partial^2L}{\partial u^j_{x^2}   \partial u^{j'}_{x^1x^2}}&\frac{\partial^2L}{\partial u^j_{x^2}   \partial u^{j'}_{2x^2}}
       \end{array}\right),$$
$$\begin{array}{ll}
A^{jj'}_{20}=\left(\begin{array}{l}\frac{\partial^2L}{\partial u^j_{2x^1}   \partial u^{j'}}\\
 \frac{\partial^2L}{\partial u^j_{x^1x^2}   \partial u^{j'}}\\
 \frac{\partial^2L}{\partial u^j_{2x^2}   \partial u^{j'}}\end{array}\right) , &
  A^{jj'}_{21}=\left(\begin{array}{ll}\frac{\partial^2L}{\partial u^j_{2x^1}   \partial u^{j'}_{x^1}}&
\frac{\partial^2L}{\partial u^j_{2x^1}   \partial u^{j'}_{x^2}}\\
\frac{\partial^2L}{\partial u^j_{x^1x^2}   \partial u^{j'}_{x^1}}&\frac{\partial^2L}{\partial u^j_{x^1x^2}   \partial u^{j'}_{x^2}}\\
\frac{\partial^2L}{\partial u^j_{2x^2}   \partial u^{j'}_{x^1}}&\frac{\partial^2L}{\partial u^j_{2x^2}   \partial u^{j'}_{x^2}}           \end{array}\right)\end{array},$$
$$A^{jj'}_{22}=\left(\begin{array}{lll}
\frac{\partial^2L}{\partial u^j_{2x^1}   \partial u^{j'}_{2x^1}}&\frac{\partial^2L}{\partial u^j_{2x^1}   \partial u^{j'}_{x^1x^2}}&\frac{\partial^2L}{\partial u^j_{2x^1}   \partial u^{j'}_{2x^2}}\\
\frac{\partial^2L}{\partial u^j_{x^1x^2}   \partial u^{j'}_{2x^1}}&\frac{\partial^2L}{\partial u^j_{x^1x^2}   \partial u^{j'}_{x^1x^2}}&\frac{\partial^2L}{\partial u^j_{x^1x^2}   \partial u^{j'}_{2x^2}}\\
\frac{\partial^2L}{\partial u^j_{2x^2}   \partial u^{j'}_{2x^1}}&\frac{\partial^2L}{\partial u^j_{2x^2}   \partial u^{j'}_{x^1x^2}}&\frac{\partial^2L}{\partial u^j_{2x^2}   \partial u^{j'}_{2x^2}}
       \end{array}\right),$$
thus
$$ A^{jj'}=\left[\begin{array}{ccc} \frac{\partial^2L}{\partial u^j   \partial u^{j'}}& \begin{array}{ll}\frac{\partial^2L}{\partial u^j   \partial u^{j'}_{x^1}}& \frac{\partial^2L}{\partial u^j   \partial u^{j'}_{x^2}}   \end{array}&\begin{array}{lll}\frac{\partial^2L}{\partial u^j   \partial u^{j'}_{2x^1}}&
\frac{\partial^2L}{\partial u^j   \partial u^{j'}_{x^1x^2}}&\frac{\partial^2L}{\partial u^j   \partial u^{j'}_{2x^2}}   \end{array}\\
 \begin{array}{l}\frac{\partial^2L}{\partial u^j_{x^1}   \partial u^{j'}}\\
 \frac{\partial^2L}{\partial u^j_{x^1}   \partial u^{j'}}\end{array}&\begin{array}{ll}\frac{\partial^2L}{\partial u^j_{x^1}   \partial u^{j'}_{x^1}}&
\frac{\partial^2L}{\partial u^j_{x^1}   \partial u^{j'}_{x^2}}\\
\frac{\partial^2L}{\partial u^j_{x^2}   \partial u^{j'}_{x^1}}&\frac{\partial^2L}{\partial u^j_{x^2}   \partial u^{j'}_{x^2}}         \end{array}&\begin{array}{lll}
\frac{\partial^2L}{\partial u^j_{x^1}   \partial u^{j'}_{2x^1}}&\frac{\partial^2L}{\partial u^j_{x^1}   \partial u^{j'}_{x^1x^2}}&\frac{\partial^2L}{\partial u^j_{x^1}   \partial u^{j'}_{2x^2}}\\
\frac{\partial^2L}{\partial u^j_{x^2}   \partial u^{j'}_{2x^1}}&\frac{\partial^2L}{\partial u^j_{x^2}   \partial u^{j'}_{x^1x^2}}&\frac{\partial^2L}{\partial u^j_{x^2}   \partial u^{j'}_{2x^2}}
       \end{array}\\
       \begin{array}{l}\frac{\partial^2L}{\partial u^j_{2x^1}   \partial u^{j'}}\\
 \frac{\partial^2L}{\partial u^j_{x^1x^2}   \partial u^{j'}}\\
 \frac{\partial^2L}{\partial u^j_{2x^2}   \partial u^{j'}}\end{array}&\begin{array}{ll}\frac{\partial^2L}{\partial u^j_{2x^1}   \partial u^{j'}_{x^1}}&
\frac{\partial^2L}{\partial u^j_{2x^1}   \partial u^{j'}_{x^2}}\\
\frac{\partial^2L}{\partial u^j_{x^1x^2}   \partial u^{j'}_{x^1}}&\frac{\partial^2L}{\partial u^j_{x^1x^2}   \partial u^{j'}_{x^2}}\\
\frac{\partial^2L}{\partial u^j_{2x^2}   \partial u^{j'}_{x^1}}&\frac{\partial^2L}{\partial u^j_{2x^2}   \partial u^{j'}_{x^2}}           \end{array}&\begin{array}{lll}
\frac{\partial^2L}{\partial u^j_{2x^1}   \partial u^{j'}_{2x^1}}&\frac{\partial^2L}{\partial u^j_{2x^1}   \partial u^{j'}_{x^1x^2}}&\frac{\partial^2L}{\partial u^j_{2x^1}   \partial u^{j'}_{2x^2}}\\
\frac{\partial^2L}{\partial u^j_{x^1x^2}   \partial u^{j'}_{2x^1}}&\frac{\partial^2L}{\partial u^j_{x^1x^2}   \partial u^{j'}_{x^1x^2}}&\frac{\partial^2L}{\partial u^j_{x^1x^2}   \partial u^{j'}_{2x^2}}\\
\frac{\partial^2L}{\partial u^j_{2x^2}   \partial u^{j'}_{2x^1}}&\frac{\partial^2L}{\partial u^j_{2x^2}   \partial u^{j'}_{x^1x^2}}&\frac{\partial^2L}{\partial u^j_{2x^2}   \partial u^{j'}_{2x^2}}
       \end{array}\end{array}\right].  $$
\end{example}

Let us recall the following formulation of the Taylor's theorem with the remainder, useful in the sequel.
\begin{theorem}[\cite{r04}]
Suppose that $f\in \mathcal{C}^2(I,\mathbb{R}),$  $a\in I,$ where $I$ is an open interval.
Then
$$f(a+t)=f(a)+f'(a)t+f''(a)\frac{t^2}{2}+e(a,t)t^2, $$
where
$$ e(a,t)=\int_{0}^{1}\left[f''(a+\sigma t)-f''(a)\right](1-\sigma)d\sigma   .$$
\end{theorem}
We can apply this theorem to rewrite (\ref{se2}) as
\begin{equation}\label{sse1}
\mathcal{F}(u+t\,\phi)=\Phi(t)=\Phi(0)+\Phi'(0)t+\frac{1}{2}t^2\Phi''(0)+e(0,t)t^2,
\end{equation}
where
\begin{equation}\label{sse2}
e(0,t)=\int_{0}^{1}\left[\Phi''(\sigma t)-\Phi''(0)\right](1-\sigma)d\sigma   .
\end{equation}

As  already shown  $\Phi(0)=\mathcal{F}(u)$ and by definition $\Phi'(0)=\delta \mathcal{F}(u,\phi).$
The quantity $\Phi''(t)$ can be found by differentiating (\ref{se10}) under the integral sign and using  the chain rule:
\begin{eqnarray}
\Phi''(t)&=&\frac{d}{d t}\int_{\Lambda}\sum_{j=1}^{m}\sum_{k=0}^{s}\sum_{h=1}^{p_k}\,\phi^j_{(k)}[h](x)\,\frac{\partial L\left(x,u^{(s)}(x)+t\, \phi^{(s)}(x)\right)}{\partial u^j_{(k)}[h]}dx\nonumber\\
&=&\int_{\Lambda}\sum_{j=1}^{m}\sum_{k=0}^{s}\sum_{h=1}^{p_k}\,\phi^j_{(k)}[h](x)\,\frac{d}{d t}\left(\frac{\partial L\left(x,u^{(s)}(x)+t\, \phi^{(s)}(x)\right)}{\partial u^j_{(k)}[h]}\right)dx\nonumber\\
&=&\int_{\Lambda}\sum_{j,j'=1}^m \sum_{k,k'=0}^s \sum_{h=1}^{p_k}\sum_{h'=1}^{p_{k'}}\phi^j_{(k)}[h]\phi^{j'}_{(k')}[h'] \frac{\partial^2 L\left(x,u^{(s)}+t\, \phi^{(s)}\right)}{\partial u^j_{(k)}[h]\,\partial u^{j'}_{(k')}[h']}dx\nonumber\\
          %  & = & \int_{\Lambda}\sum_{j,j'=1}^m \sum_{k,k'=0}^s \,\phi^j_{(k)}(x)\,  A^{jj'}_{kk'}\left(x,u^{(s)}(x)+t\, \phi^{(s)}(x)\right)  \,^{T}\phi^{j'}_{(k')}(x)\,dx\nonumber\\
          % & = & \int_{\Lambda}\sum_{j,j'=1}^m\,\phi^{j(s)}(x)\, A^{jj'}\left(x,u^{(s)}(x)+t\, \phi^{(s)}(x)\right)  \,^{T}\phi^{j'(s)}\,dx\nonumber\\
            & = & \int_{\Lambda}\phi^{(s)}(x)\, A\left(x,u^{(s)}(x)+t\, \phi^{(s)}(x)\right) \, ^{T}\phi^{(s)}(x)\,dx,\label{sse3}
\end{eqnarray}
where the notation $^{T}(\cdot)$ denotes the transpose of $(\cdot)$. In particular at $t=0,$ we get
\begin{eqnarray}
\Phi''(0)&=&\int_{\Lambda}\sum_{j,j'=1}^m \sum_{k,k'=0}^s \sum_{h=1}^{p_k}\sum_{h'=1}^{p_{k'}}\phi^j_{(k)}[h](x)\phi^{j'}_{(k')}[h'](x) \frac{\partial^2 L\left(x,u^{(s)}(x)\right)}{\partial u^j_{(k)}[h]\,\partial u^{j'}_{(k')}[h']}dx\nonumber\\
& = & \int_{\Lambda}\phi^{(s)}(x)\, A\left(x,u^{(s)}(x)\right) \, ^{T}\phi^{(s)}(x)\,dx.\label{sse4}
\end{eqnarray}
We then arrive at the following formulation.
\begin{definition}\label{df2}
 The second variation of the functional $\mathcal{F}$ in the neighborhood of  $u$ in the direction $\phi$ is defined by
 \begin{equation}\label{sse5}
 \delta^2 \mathcal{F}(u+t\, \phi,\phi)=\Phi''(t) .
 \end{equation}
 In particular, the second variation of $\mathcal{F}$ at $u$ in the direction $\phi$ is given by
 \begin{equation}\label{sse6}
 \delta^2 \mathcal{F}(u,\phi)=\Phi''(0) .
 \end{equation}
\end{definition}
The Taylor expansion  (\ref{sse1}) can be now re-expressed as
\begin{eqnarray}
\Phi(t)&=&\mathcal{F}(u+t\,\phi)\,\,=\,\, \mathcal{F}(u)+t\delta \mathcal{F}(u,\phi)+\frac{1}{2}t^2\delta^2 \mathcal{F}(u,\phi)\nonumber\\
&+& t^2  \int_{0}^{1}\left[\delta^2 \mathcal{F}(u+\sigma t\,\phi,\phi)-\delta^2 \mathcal{F}(u,\phi)\right](1-\sigma)d\sigma   .\label{sse7}
\end{eqnarray}
Let us also recall the following two  results which are important to prove the main results of this work.
\begin{lemma}[\cite{r04}]\label{lem2}
 Suppose that $A=\left(a_{ij}\right)$ is an $N\times N$ matrix, and set $|||  A    |||=\sqrt{\sum_{i,j=1}^N  a_{ij}^2 }.$ Then
 ( $\| \cdot \|$ denotes the Euclidean norm $\mathbb{R}^N$)
 \begin{enumerate}
 \item $v\cdot w\leq \| v \|  \| w  \|$ for all $v$ and $w$ in $\mathbb{R}^N.$
 \item $|A v |\leq |||  A    |||\, \|v \|$ for all $v$ in $\mathbb{R}^N.$
 \item $| v\cdot A w  | \leq |||  A    |||\, \| v \|  \| w  \|   $ for all $v$ and $w$ in $\mathbb{R}^N.$
 \end{enumerate}
\end{lemma}
\begin{definition}[Positive semi-definite]
A symmetric matrix $A\in \mathbb{R}^{N^2}$ is called positive semi-definite if $v\cdot A\,v\geq 0$ for all $v\in\mathbb{R}^N.$
\end{definition}
\begin{definition}[Positive definite]
A symmetric matrix $A\in \mathbb{R}^{N^2}$ is called positive definite if $v\cdot A\,v> 0$ for all $v\in\mathbb{R}^N\setminus \{ 0 \}.$
\end{definition}
\begin{lemma}[\cite{r04}]\label{lem3}
 Suppose that $A$ is a positive definite $N\times N$ matrix. Then there is a constant $k>0$ such that
 $v\cdot A \,v\geq k \| v\|^2$ for all $v$ in $\mathbb{R}^N.$
\end{lemma}

There results the following.
\begin{lemma}\label{lem4}
 Suppose that $L\in \mathcal{C}^2(\Lambda\times \Omega,\mathbb{R}),$ and $u\in \mathcal{C}_b^s(\Lambda,U)$ is admissible for $\mathcal{F}.$
 Then for any $\epsilon > 0 ,$ there is $\delta > 0$ such that
 $$ | e(0,t)  |\leq \frac{\epsilon}{2} \int_{\Lambda}\sum_{j=1}^{m}\sum_{k=0}^{s}\sum_{h=1}^{p_k}\,\left|\phi^j_{(k)}[h](x) \right|^2 dx
 = \frac{\epsilon}{2} \int_{\Lambda} \left\|\phi^{(s)}(x) \right\|^2 dx $$
for all $\phi\in \mathcal{C}^{\infty}_{0}(\Lambda,U)$ and $| t|\leq \epsilon_0$ such that
\begin{itemize}
\item[(i)] $u+t\,\phi$ is  admissible for $\mathcal{F},$
 \item[(ii)] $| t|| \sigma| \left\|\phi^{(s)}(x) \right\|<\delta$
for all $x\in \Lambda$  and $0\leq \sigma \leq 1.$
\end{itemize}
Here, $\| \cdot \|$ denotes the Euclidean norm in $\mathbb{R}^{q_s}.$
\end{lemma}
\begin{proof}
 Using (\ref{sse3}) and (\ref{sse4}), we see that
 \begin{eqnarray}
  e(0,t)  & = & \int_{0}^{1}(1-\sigma)\left[\delta^2 \mathcal{F}(u+\sigma t\,\phi,\phi)-\delta^2 \mathcal{F}(u,\phi)\right]d\sigma\nonumber\\
  &  = & \int_{0}^{1}(1-\sigma) \int_{\Lambda}\phi^{(s)}(x)\,\left[ A\left(x,u^{(s)}(x)+\sigma\,t\, \phi^{(s)}(x)\right)\right.\nonumber\\
 &-&\left.A\left(x,u^{(s)}(x)\right)\right]\, ^{T}\phi^{(s)}(x) dx d\sigma.\nonumber
 \end{eqnarray}
 Each entry of the matrix $A$ is a second derivative of the function $L$ with respect to the coordinates in $\Omega.$

Let $\epsilon > 0.$ Since all the second derivatives of $L$ with respect to the variables in $\Omega$ are continuous, then the matrix
$A$ is continuous and there exists $\delta > 0$ such that, for all $\phi\in \mathcal{C}^{\infty}_{0}(\Lambda,U),$
$$\left\|\left(u^{(s)}(x) +t \,\sigma\,\phi^{(s)}(x)\right)-u^{(s)}(x) \right\|=| t|| \sigma|
\left\|\phi^{(s)}(x) \right\|<\delta \Longrightarrow |||B(x,\sigma, t)||| <\epsilon$$
for all $x\in \Lambda,$ $| t|\leq \epsilon_0$  and $0\leq \sigma \leq 1,$ where
$$  B(x,\sigma, t)= A\left(x,u^{(s)}(x)+\sigma\,t\, \phi^{(s)}(x)\right) -A\left(x,u^{(s)}(x)\right) .  $$
Therefore, using the  continuity of the bilinear form induced by the matrix $A$  (see the third property  of  \textrm{Lemma} \ref{lem2}),
 we obtain the required result:
\begin{eqnarray}
  |e(0,t) | & \leq & \epsilon \int_{0}^{1}(1-\sigma)d\sigma \int_{\Lambda}\sum_{j=1}^{m}\sum_{k=0}^{s}\sum_{h=1}^{p_k}\,\left|\phi^j_{(k)}[h](x) \right|^2 dx\nonumber\\
  &  \leq & \frac{\epsilon}{2} \int_{\Lambda}\sum_{j=1}^{m}\sum_{k=0}^{s}\sum_{h=1}^{p_k}\,\left|\phi^j_{(k)}[h](x) \right|^2 dx\,\,=\,\,\frac{\epsilon}{2} \int_{\Lambda} \left\|\phi^{(s)}(x) \right\|^2 dx. \nonumber
 \end{eqnarray}
Here, $\|  \cdot\|$ denotes the Euclidean norm in $\mathbb{R}^{q_s}$ since $\phi^{(s)}(x)\in \mathbb{R}^{q_s}$ for all $x\in \Lambda.$
\end{proof}
%\newpage
\begin{theorem}\label{thm}
 Suppose that $L\in \mathcal{C}^2(\Lambda\times \Omega,\mathbb{R}),$ and $u\in \mathcal{C}_b^s(\Lambda,U)$ is admissible for $\mathcal{F}.$
 \begin{enumerate}
  \item If $u$ is a weak local minimum for $\mathcal{F}, $ then for all $\phi\in \mathcal{C}^{\infty}_{0}(\Lambda,U),$ we have
  \begin{equation}\label{sse8}
       \delta \mathcal{F}(u,\phi)=0 \quad \emph{\emph{and}} \quad  \delta^2 \mathcal{F}(u,\phi) \geq 0.
  \end{equation}
  \item If $u$ is a weak local extremum for $\mathcal{F}$ and there is a constant $k>0$ such that
  \begin{equation}\label{sse9}
       \delta^2 \mathcal{F}(u,\phi)\geq k \int_{\Lambda} \left\|\phi^{(s)}(x) \right\|^2 dx
  \end{equation}
  for all $\phi\in \mathcal{C}^{\infty}_{0}(\Lambda,U),$
  then $u$ is a strict weak local minimum.
 \end{enumerate}
\end{theorem}
\begin{proof}
 For the first part  of \textrm{Theorem} \ref{thm}, the assumption that $u$ is a weak local minimum for $\mathcal{F}$
 implies that $t=0$ is a local minimum for the function  $\Phi(t)=\mathcal{F}(u+t\,\phi).$ Consequently,
 $0=\Phi'(0)=\delta \mathcal{F}(u,\phi).$
 The Taylor expansion (\ref{sse1}) of $\Phi$ gives
 $$ \Phi''(0)=2\,\frac{\Phi(t)-\Phi(0)}{t^2} +2\,e(0,t)  $$
 which leads to
 $$ 0\leq  \lim_{t\rightarrow 0} \,2\,\frac{\Phi(t)-\Phi(0)}{t^2} = \Phi''(0)=\delta^2 \mathcal{F}(u,\phi)    $$
since  $ \lim_{t\rightarrow 0}\,e(0,t)=0 $ and $\Phi(t)\geq \Phi(0)$ for all $t\neq 0.$

For the second part, we suppose that $v\in \mathcal{C}_b^s(\Lambda,U)$ is admissible for $\mathcal{F}.$
Let $\phi\in \mathcal{C}^{\infty}_{0}(\Lambda,U)$ so that $v=u+t\,\phi$ for some  $t\in \mathbb{R}$ such that $| t|\leq \epsilon_0.$ Then
$$
\mathcal{F}(v)=\mathcal{F}(u+t\,\phi)=\mathcal{F}(u)+t\delta \mathcal{F}(u,\phi)+\frac{1}{2}t^2\delta^2 \mathcal{F}(u,\phi)+t^2 \,e(0,t),$$
where
$$e(0,t)=\int_{0}^{1}\left[\delta^2 \mathcal{F}(u+\sigma t\,\phi,\phi)-\delta^2 \mathcal{F}(u,\phi)\right](1-\sigma)d\sigma   .$$
By assumption, $u$ is a weak extremum, so $\delta \mathcal{F}(u,\phi)=0.$
By \textrm{Lemma} \ref{lem4}, there is $\epsilon > 0$ such that
$$ | e(0,t)  |\leq  \frac{k}{4} \int_{\Lambda} \left\|\phi^{(s)}(x) \right\|^2 dx $$
provided  $| t|| \sigma| \left\|\phi^{(s)}(x) \right\|<\epsilon$
for all $x\in \Lambda$  and $0\leq \sigma \leq 1.$
Therefore, using (\ref{sse9}), if $\| v-u  \|_s=| t|\,\| \phi\|_s < \epsilon$ we have
\begin{eqnarray}
 \mathcal{F}(v)&\geq& \mathcal{F}(u)+\frac{t^2}{2}\delta^2 \mathcal{F}(u,\phi)-t^2| e(0,t)  |\nonumber\\
           &   \geq &  \mathcal{F}(u)+   \frac{k\,t^2}{2}  \int_{\Lambda} \left\|\phi^{(s)}(x) \right\|^2 dx-\frac{k\,t^2}{4}\int_{\Lambda} \left\|\phi^{(s)}(x) \right\|^2 dx\nonumber\\
           & = & \mathcal{F}(u)+  \frac{k\,t^2}{4}  \int_{\Lambda} \left\|\phi^{(s)}(x) \right\|^2 dx.\nonumber
\end{eqnarray}
If $v\neq u,$ i.e. $\phi \neq 0,$ then the integral on the right hand side is strictly positive, and we  have $\mathcal{F}(v)> \mathcal{F}(u).$ Therefore $u$ is a strict weak
local minimum.
\end{proof}

\begin{theorem}\label{thmd}
 Let $\Lambda$ be a bounded connected  subset of $X.$ Suppose that $L\in \mathcal{C}^2(\Lambda\times \Omega,\mathbb{R}),$ and $u\in \mathcal{C}_b^s(\Lambda,U)$ is admissible for $\mathcal{F}.$
  If $u$ is a weak local extremum for $\mathcal{F}$ and
  \begin{equation}\label{6sse9}
        \delta^2 \mathcal{F}(u,\phi) \geq 0
  \end{equation}
  for all $\phi\in \mathcal{C}^{\infty}_{0}(\Lambda,U),$
  then $u$ is a weak local minimum.
\end{theorem}
\begin{proof}
Let $\widehat{\epsilon}> 0.$
Suppose that $v\in \mathcal{C}_b^s(\Lambda,U)$ is admissible for $\mathcal{F}.$
Let $\phi\in \mathcal{C}^{\infty}_{0}(\Lambda,U)$ so that $v=u+t\,\phi$ for some  $t\in \mathbb{R}$ such that $| t|\leq \epsilon_0.$ Then
$$
\mathcal{F}(v)=\mathcal{F}(u+t\,\phi)=\mathcal{F}(u)+t\delta \mathcal{F}(u,\phi)+\frac{1}{2}t^2\delta^2 \mathcal{F}(u,\phi)+t^2 \,e(0,t),$$
where
$$e(0,t)=\int_{0}^{1}\left[\delta^2 \mathcal{F}(u+\sigma t\,\phi,\phi)-\delta^2 \mathcal{F}(u,\phi)\right](1-\sigma)d\sigma   .$$
By assumption, $u$ is a weak extremum, so $\delta \mathcal{F}(u,\phi)=0.$
By \textrm{Lemma} \ref{lem4}, there is $\epsilon > 0$ such that
$$ | e(0,t)  |\leq  \frac{\widehat{\epsilon}}{2} \int_{\Lambda} \left\|\phi^{(s)}(x) \right\|^2 dx\leq \frac{\widehat{\epsilon}}{2}\, \emph{\emph{mes}}(\Lambda)\,\| \phi\|_s^2 $$
provided  $| t|| \sigma| \left\|\phi^{(s)}(x) \right\|<\epsilon$
for all $x\in \Lambda$  and $0\leq \sigma \leq 1.$
Therefore, using (\ref{6sse9}), if $\| v-u  \|_s=| t|\,\| \phi\|_s < \epsilon$ with $\phi\neq 0,$ we have
\begin{eqnarray}
 \mathcal{F}(v)&\geq& \mathcal{F}(u)+\frac{t^2}{2}\delta^2 \mathcal{F}(u,\phi)-t^2| e(0,t)  |\nonumber\\
           &   \geq &  \mathcal{F}(u)-\frac{\widehat{\epsilon}}{2} \,t^2\, \emph{\emph{mes}}(\Lambda)\left\|\phi \right\|_s^2\nonumber\\
           & \geq & \mathcal{F}(u)-\frac{\widehat{\epsilon}}{2} \,\epsilon^2 \,\emph{\emph{mes}}(\Lambda).\nonumber
\end{eqnarray}
Thus,
\begin{eqnarray}
 \mathcal{F}(v)  & \geq & \mathcal{F}(u)-\lim_{\widehat{\epsilon}\rightarrow 0}\left[\frac{\widehat{\epsilon}}{2}  \,\epsilon^2\, \emph{\emph{mes}}(\Lambda)\right]=\mathcal{F}(u).\nonumber
\end{eqnarray}
We  have $\mathcal{F}(v)\geq \mathcal{F}(u).$ Therefore $u$ is a  weak
local minimum.
\end{proof}

In part (1) of \textrm{Theorem} \ref{thm}, the fact that the second variations must be nonnegative is a necessary condition for $u$ to be
a local minimum.
%%%%%%%%%%%%%%%%%%%%%%%%%%%%%%%%%%%%%%%%%%%%%%%%%%%%%%%%%%%%%%%%%%%%%%%%%%%%%%%
\subsection{Legendre necessary conditions}

According to \textrm{Theorem} \ref{thm},  if $u$ is a weak local minimum
for the functional $\mathcal{F},$ then       $\delta^2 \mathcal{F}(u,\phi) \geq 0$ for all
$\phi\in \mathcal{C}^{\infty}_0(\Lambda,U).$ Here we find some natural and nontrivial consequences of that condition.

 Construct nonzero functions $\psi_l$ by $\psi_0=1$ and for $l=1,2,  \cdots  ,s$
$$ \psi_l(y)=\left\{\begin{array}{lr}
0 & -\infty<y\leq -1\\
1-y^l \,\emph{\emph{sign}}(y) & -1\leq y\leq +1\\
0 & +1\leq y < +\infty
   \end{array}\right.$$
if $l$ is odd, and
$$ \psi_l(y)=\left\{\begin{array}{lr}
0 & -\infty<y\leq -1\\
1-y^l  & -1\leq y\leq +1\\
0 & +1\leq y < +\infty
   \end{array}\right.$$
if $l$ is even.

It is clear that $\psi_l\in \mathcal{C}^{\infty}\left(\mathbb{R}\setminus \{-1,1  \},\mathbb{R}\right)$
and satisfy $\psi_l(y)=0$ for all $y$ with $|y|>1,$
i.e. $\psi_l\in \mathcal{C}_0^{\infty}\left(\mathbb{R}\setminus \{-1,1  \},\mathbb{R}\right).$ Furthermore,
$\psi_l\in \mathcal{C}(\mathbb{R},\mathbb{R})$ with $\psi_l(-1)=\psi_l(1)=0.$ We also have
\begin{equation}\label{esc1}
 (\psi_l)_{(l)}[1](y)\in\{ -l!,0,+l! \}   \quad \forall\, y\in\mathbb{R}
\end{equation}
that is $\left(\psi_l\right)_{(l)}[1] $ is constant on $\mathbb{R}.$
Thus,
\begin{equation}\label{esc2}
 (\psi_l)_{(l+\nu)}[1](y)=0\quad \forall\,y\in\mathbb{R} ,\,\,  \nu\geq 1.
\end{equation}
Let $x_0=\left( x_0^1, \cdots   ,x_0^n   \right)\in \Lambda.$ Since $\Lambda$
is an open subset of $\mathbb{R}^n,$ there is $r_0>0$ such that
$B(x_0,r_0)=\{  x\in \mathbb{R}^n\,:\, \| x-x_0  \|<r_0   \}\subset \Lambda.$
Let $\xi=\left( \xi^1, \cdots  ,\xi^m    \right)\in \mathbb{R}^m$ and $0<\epsilon < \frac{r_0}{\sqrt{n}}.$
Set $\phi_l(x)=\left( \phi_l^1(x) ,    \cdots ,\phi_l^m(x) \right)\cdot\mathcal{X}_{B(x_0,r_0)}(x)$ where
$\mathcal{X}_{B(x_0,r_0)}$ is the characteristic function of the set $B(x_0,r_0)$ and
\begin{equation}\label{esc3}
 \phi_l^j(x)=\xi^j\epsilon^l\sum_{i=1}^n\psi_l\left(  \frac{x^i-x^i_0}{\epsilon} \right).
 \end{equation}
Clearly,  the support of $\phi_l$ is a compact contained in $\Lambda.$
We have for all $k\in \mathbb{N}$ and $h=1,2, \cdots  ,p_k$
\begin{equation}\label{esc5}
 \left( \phi^j_l\right)_{(k)}[h](x)=\xi^j\epsilon^{l-k}\left( \psi_l \right)_{(k)}[1]\left( \frac{x^{i(h)}-x_0^{i(h)}}{\epsilon} \right)
\end{equation}
for some  $i(h)\in\{ 1,2, \cdots ,n  \}.$
Using (\ref{esc1}) and (\ref{esc2}), we see that
\begin{equation}\label{esc7}
 \left( \phi^j_l\right)_{(l)}[h]\in \{ -\xi^j\,l!,0,+\xi^j\,l! \} \quad \forall\, h=1,2, \cdots  ,p_l
\end{equation}
and for $\nu\geq 1$
\begin{equation}\label{esc6}
 \left( \phi^j_l\right)_{(l+\nu)}[h]=0\quad \forall\, h=1,2, \cdots  ,p_{l+\nu}.
\end{equation}
Therefore, if $u$ is a weak local minimum for the functional $\mathcal{F},$  we have
for all $l=0,1,2,\cdots ,s$
\begin{equation}\label{esc8}
0\leq \delta^2 \mathcal{F}(u,\phi_l) =I_1+I_2+I_3,
\end{equation}
where
\begin{eqnarray}
I_1&=&\int_{\|x-x_0  \|\leq r_0}\sum_{j,j'=1}^{m}\sum_{k,k'=l+1}^{s} \sum_{h=1}^{p_k}\sum_{h'=1}^{p_{k'}} \nonumber\\
& &\left( \phi^j_l\right)_{(k)}[h](x)\left( \phi^{j'}_l\right)_{(k')}[h'](x)  \frac{\partial^2 L\left(x,u^{(s)}(x)\right)}{\partial u^j_{(k)}[h]  \partial u^{j'}_{(k')}[h']}  \,dx ;\label{esc9}
\end{eqnarray}
\begin{equation}\label{esc10}
I_2=\int_{\|x-x_0  \|\leq r_0}\sum_{j,j'=1}^{m}\sum_{h,h'=1}^{p_l}\left( \phi^j_l\right)_{(l)}[h](x)\left( \phi^{j'}_l\right)_{(l)}[h'](x)  \frac{\partial^2 L\left(x,u^{(s)}(x)\right)}{\partial u^j_{(l)}[h]  \partial u^{j'}_{(l)}[h']}  \,dx;
\end{equation}
\begin{eqnarray}
I_3&=&2\int_{\|x-x_0  \|\leq r_0}\sum_{j,j'=1}^{m}\sum_{k=0}^{l-1}\sum_{k'=k+1}^{l} \sum_{h=1}^{p_k}\sum_{h'=1}^{p_{k'}} \nonumber\\
& &\left( \phi^j_l\right)_{(k)}[h](x)\left( \phi^{j'}_l\right)_{(k')}[h'](x)  \frac{\partial^2 L\left(x,u^{(s)}(x)\right)}{\partial u^j_{(k)}[h]  \partial u^{j'}_{(k')}[h']}  \,dx .\label{esc11}
\end{eqnarray}
Of course, if $l=0$ there is not the integral $I_3.$ If $s=0$    $I_1$ and $I_3$ do not exist.
By (\ref{esc6}), we see that $I_1=0.$ Using (\ref{esc3}) and (\ref{esc5}) in (\ref{esc10}), we have
\begin{eqnarray}
I_2&=&\int_{\|x-x_0  \|\leq r_0}\sum_{j,j'=1}^{m}\xi^j\xi^{j'}\sum_{h,h'=1}^{p_l} \frac{\partial^2 L\left(x,u^{(s)}(x)\right)}{\partial u^j_{(l)}[h]  \partial u^{j'}_{(l)}[h']} \nonumber\\
& \times &  \left(\psi_l\right)_{(l)}[1]\left( \frac{x^{i(h)}-x_0^{i(h)}}{\epsilon} \right) \,\left(\psi_l\right)_{(l)}[1]\left( \frac{x^{i(h')}-x_0^{i(h')}}{\epsilon} \right)      dx .\label{esc12}
\end{eqnarray}
If we set $x=x_0+\epsilon\, y$ then $dx=\epsilon^n\,dy$ and using (\ref{esc1}), $I_2$ satisfies
\begin{equation}\label{esc13}
I_2\leq \epsilon^n (l!)^2 \int_{\|y  \|\leq  1} \sum_{j,j'=1}^{m} \xi^j\xi^{j'}\left(\sum_{h,h'=1}^{p_l} \frac{\partial^2 L\left(x_0+\epsilon \,y,u^{(s)}(x_0+\epsilon \,y)\right)}{\partial u^j_{(l)}[h]  \partial u^{j'}_{(l)}[h']}   \right)             dy.
\end{equation}
In a similar way, $I_3$ becomes
\begin{eqnarray}
I_3&=&2 \epsilon^n \,\epsilon^{2l-(k+k')}   \int_{\|y  \|\leq1}\sum_{j,j'=1}^{m}\xi^j\xi^{j'}\left(                 \sum_{k=0}^{l-1}\sum_{k'=k+1}^{l} \sum_{h=1}^{p_k}\sum_{h'=1}^{p_{k'}} \left(\psi_l\right)_{(k)}[1]\left( y^{i(h)} \right)\right.\nonumber\\
& & \times\left.\left(\psi_l\right)_{(k')}[1]\left( y^{i(h')} \right)
\frac{\partial^2 L\left(x_0+\epsilon \,y,u^{(s)}(x_0+\epsilon \,y)\right)}{\partial u^j_{(k)}[h]  \partial u^{j'}_{(k')}[h']}  \right)dy.\label{esc14}
\end{eqnarray}
Substituting (\ref{esc13}) and (\ref{esc14}) into (\ref{esc8}), we have
\begin{eqnarray}
0&\leq &\sum_{j,j'=1}^{m}\xi^j\xi^{j'}\left[ (l!)^2 \int_{\|y  \|\leq  1} \sum_{h,h'=1}^{p_l} \frac{\partial^2 L\left(x_0+\epsilon \,y,u^{(s)}(x_0+\epsilon \,y)\right)}{\partial u^j_{(l)}[h]  \partial u^{j'}_{(l)}[h']}              dy   \right.\nonumber\\
& + & \epsilon^{2l-(k+k')} \int_{\|y  \|\leq  1} \sum_{k=0}^{l-1}\sum_{k'=k+1}^{l} \sum_{h=1}^{p_k}\sum_{h'=1}^{p_{k'}} \left(\psi_l\right)_{(k)}[1]\left( y^{i(h)} \right)\nonumber\\
&\times &\left.\left(\psi_l\right)_{(k')}[1]\left( y^{i(h')} \right)
\frac{\partial^2 L\left(x_0+\epsilon \,y,u^{(s)}(x_0+\epsilon \,y)\right)}{\partial u^j_{(k)}[h]  \partial u^{j'}_{(k')}[h']} \,dy\right].
\label{esc15}
\end{eqnarray}
We have $\epsilon^{2l-(k+k')}\longrightarrow 0$ as $\epsilon \rightarrow 0$ since in $I_3,$ $2l-(k+k')\geq 1.$\\
 Therefore,
as $\epsilon \rightarrow 0,$ the second term in (\ref{esc15}) vanishes and it remains
\begin{equation}\label{esc16}
0 \leq  (l!)^2 \sum_{j,j'=1}^{m} \xi^j\xi^{j'}  \left(\sum_{h,h'=1}^{p_l} \frac{\partial^2 L\left(x_0,u^{(s)}(x_0)\right)}{\partial u^j_{(l)}[h]  \partial u^{j'}_{(l)}[h']}   \right)        \int_{\|y  \|\leq  1}     dy
\end{equation}
from which we deduce
\begin{equation}\label{esc17}
\sum_{j,j'=1}^{m}  \left(\sum_{h,h'=1}^{p_l} \frac{\partial^2 L\left(x_0,u^{(s)}(x_0)\right)}{\partial u^j_{(l)}[h]  \partial u^{j'}_{(l)}[h']}   \right)     \xi^j\xi^{j'}   \geq 0.
\end{equation}
Since $x_0\in \Lambda$ and  $\xi \in \mathbb{R}^m$ are arbitrary, we have proved the following theorem
 \begin{theorem}\label{thmm}
 Suppose that $L\in \mathcal{C}^2(\Lambda\times \Omega,\mathbb{R}),$ and $\overline{u}$ is a weak local minimum for $\mathcal{F}.$
Then for all $x\in \Lambda$ and $\xi=\left(\xi^1,  \cdots , \xi^m   \right)\in \mathbb{R}^m,$
\begin{equation}\label{esc18}
 \sum_{j,j'=1}^{m}  \left(\sum_{h,h'=1}^{p_l} \frac{\partial^2 L\left(x,\overline{u}^{(s)}(x)\right)}{\partial u^j_{(l)}[h]
 \partial u^{j'}_{(l)}[h']}   \right)     \xi^j\xi^{j'}   \geq 0 \qquad l=0,1,2,\cdots,s,
\end{equation}
i.e. for all $x\in \Lambda,$ the square matrices $A^{jj'}_{ll}\left(x,\overline{u}^{(s)}(x)\right),$ $l=0,1,2,\cdots,s,$
are positive semi-definite.
 \end{theorem}
The inequalities in (\ref{esc18}) are called the general forms of Legendre conditions. They define by  \textrm{Theorem} \ref{thmm}
 new necessary conditions for $\overline{u}$ to be a weak local minimum of $\mathcal{F}.$ We  say that the function $\overline{u}$
satisfies the strict Legendre conditions if the matrices $A^{jj'}_{ll}\left(x,\overline{u}^{(s)}(x)\right),$ $l=0,1,2,\cdots,s,$
 are positive definite, uniformly for all $x\in \Lambda.$
 %%%%%%%%%%%%%%%%%%%%%%%%%%%%%%%%%%%%%%%%%%%%%%%%%%%%%%%%%%%%%%%%%%%%%%%%%%%%%%%%%%%%%%%%%%%%%%%%
\subsection{Relevant sufficient  conditions}

Part (2) of \textrm{Theorem} \ref{thm} gives us a sufficient condition for a function to be a minimum. However,
the conditions involving the second variations are not easy to satisfy. So, the  results of this subsection are useful as they
 imply the condition (\ref{sse9}).
 \begin{theorem}
 Suppose that $L\in \mathcal{C}^2(\Lambda\times \Omega,\mathbb{R}),$ and $u$ is a weak local extremum for $\mathcal{F}.$
 If the matrix     $A\left(x,u^{(s)}(x)\right)$ defined by (\ref{sse0}) is positive definite for all $x\in \Lambda,$
  then $u$ is a strict weak local minimum.
 \end{theorem}
 \begin{proof}
 By (\ref{sse4}), (\ref{sse6}) and \textrm{Lemma} \ref{lem3}, for all $\phi\in \mathcal{C}^{\infty}_{0}(\Lambda,U)$ we have
 $$\delta^2 \mathcal{F}(u,\phi)=\int_{\Lambda}\phi^{(s)}(x)\, A\left(x,u^{(s)}(x)\right) \, ^{T}\phi^{(s)}(x)\,dx  \geq k \int_{\Lambda} \left\|\phi^{(s)}(x) \right\|^2 dx $$
 for some $k>0.$ By part (2) of \textrm{Theorem} \ref{thm},  $u$ is a strict weak  minimum for $\mathcal{F}.$
 \end{proof}

 \begin{theorem}
Let $\Lambda$ be a bounded connected  subset of $X.$ Suppose that $L\in \mathcal{C}^2(\Lambda\times \Omega,\mathbb{R}),$ and $u$ is a weak local extremum for $\mathcal{F}.$
 If the matrix     $A\left(x,u^{(s)}(x)\right)$ defined by (\ref{sse0}) is semi-positive definite for all $x\in \Lambda,$
  then $u$ is a  weak local minimum.
 \end{theorem}
 \begin{proof}
 By   hypothesis,
 the function $ V(x;\phi)=\phi^{(s)}(x)\, A\left(x,u^{(s)}(x)\right) \, ^{T}\phi^{(s)}(x) $ is continuous and positive on $\Lambda$
 for all $\phi\in \mathcal{C}^{\infty}_{0}(\Lambda,U).$
 Therefore, by (\ref{sse4}) and (\ref{sse6}), for all $\phi\in \mathcal{C}^{\infty}_{0}(\Lambda,U),$ we have
 $$\delta^2 \mathcal{F}(u,\phi)=\int_{\Lambda}\phi^{(s)}(x)\, A\left(x,u^{(s)}(x)\right) \, ^{T}\phi^{(s)}(x)\,dx  \geq 0. $$
  Thus, by  \textrm{Theorem} \ref{thmd},  $u$ is a  weak  minimum for $\mathcal{F}.$
 \end{proof}
$ $\\

The second variation of $\mathcal{F}$ is given by
\begin{eqnarray}
\delta^2 \mathcal{F}(u,\phi)&=&\int_{\Lambda} \sum_{j,j'=1}^{m} \sum_{k,k'=0}^{s} \phi^j_{(k)}(x)\, A^{jj'}_{kk'}\left(x,u^{(s)}(x)\right) \, {}^{T}\phi^{j'}_{(k')}(x)      \,dx\nonumber\\
&=& \int_{\Lambda} \sum_{j,j'=1}^{m} \sum_{k=0}^{s} \left[ \phi^j_{(k)} \,A^{jj'}_{kk}\, {}^{T}\phi^{j'}_{(k)}  +2\sum_{{}^{k'=0}_{k'\neq k}}^{s} \phi^j_{(k)} \,A^{jj'}_{kk'} \,{}^{T}\phi^{j'}_{(k')}       \right] dx               \nonumber\\
 &=&I_1+2\,I_2, \label{esd1}
\end{eqnarray}
where  the  matrices $A^{jj'}_{kk'}\left(x,u^{(s)}(x)\right)$ are defined by (\ref{sse00}) and
\begin{equation}\label{esd2}
I_1=\int_{\Lambda} \sum_{j,j'=1}^{m} \sum_{k=0}^{s}  \phi^j_{(k)} \,A^{jj'}_{kk}\, {}^{T}\phi^{j'}_{(k)}\,dx;
\end{equation}
\begin{equation}\label{esd3}
I_2=\int_{\Lambda} \sum_{j,j'=1}^{m} \sum_{k=0}^{s} \sum_{{}^{k'=0}_{k'\neq k}}^{s} \phi^j_{(k)} \,A^{jj'}_{kk'} \,{}^{T}\phi^{j'}_{(k')} \,dx.
\end{equation}
Integral $I_1$ can be rewritten as
\begin{eqnarray}
I_1&=&\int_{\Lambda} \sum_{j=1}^{m}\left[ \sum_{k=0}^{s}  \phi^j_{(k)} \,A^{jj}_{kk}\, {}^{T}\phi^{j}_{(k)}+2 \sum_{{}^{j'=1}_{j'\neq j}}^{m}    \sum_{k=0}^{s}  \phi^j_{(k)} \,A^{jj'}_{kk}\, {}^{T}\phi^{j'}_{(k)}\right]dx\nonumber\\
& =&J_1+2\,J_2,
\end{eqnarray}
\begin{equation}\label{esd4}
J_1=\int_{\Lambda} \sum_{j=1}^{m} \sum_{k=0}^{s}  \phi^j_{(k)} \,A^{jj}_{kk}\, {}^{T}\phi^{j}_{(k)}\,dx;
\end{equation}
\begin{equation}\label{esd5}
J_2=\int_{\Lambda} \sum_{j=1}^{m} \sum_{{}^{j'=1}_{j'\neq j}}^{m}    \sum_{k=0}^{s}  \phi^j_{(k)} \,A^{jj'}_{kk}\, {}^{T}\phi^{j'}_{(k)}  \,dx.
\end{equation}
Thus, the second variation can be written as
\begin{equation}
\delta^2 \mathcal{F}(u,\phi)=I_1+2\,I_2=J_1+2\,J_2+2\,I_2=J_1+2\,(J_2+I_2).
\end{equation}
We can now  prove the following new sufficient condition.
\begin{theorem}\label{thmmm}
 Suppose that $L\in \mathcal{C}^2(\Lambda\times \Omega,\mathbb{R}),$ and $u$ is a weak local extrema for $\mathcal{F}.$
 If
 \begin{itemize}
 \item[(i)] $J_2+I_2\geq 0,$ and
  \item[(ii)] the square matrices $A^{jj'}_{kk}\left(x,u^{(s)}(x)\right)$  are positive definite for all $x\in \Lambda,$ i.e., satisfy the strict Legendre conditions,
 \end{itemize}
  then $u$ is a strict weak local minimum for $\mathcal{F}.$
 \end{theorem}
 \begin{proof}
 We have shown that
 \begin{equation}\label{esd6}
\delta^2 \mathcal{F}(u,\phi)=J_1+2\,(J_2+I_2),
\end{equation}
 where $J_1,$ $J_2$ and $I_2$ are defined by (\ref{esd4}), (\ref{esd5}) and (\ref{esd3}), respectively.
 By condition (ii), using the \textrm{Lemma} \ref{lem3}, there exist constants $\alpha^j_k>0$ such that
 \begin{eqnarray}
J_1&\geq&   \int_{\Lambda} \sum_{j=1}^{m} \sum_{k=0}^{s} \alpha^j_k \,\,\phi^j_{(k)} \, {}^{T}\phi^{j}_{(k)}\,dx       \nonumber\\
& \geq&    \alpha  \int_{\Lambda} \sum_{j=1}^{m} \sum_{k=0}^{s}  \sum_{h=1}^{p_k}  \left(\phi^j_{(k)}[h](x)\right)^2 dx       \, = \,\alpha \int_{\Lambda} \left\|\phi^{(s)}(x) \right\|^2 dx ,  \label{esd7}
\end{eqnarray}
 where $0<\alpha=\min\left\{ \alpha^j_k, \,1\leq j\leq m,\,0\leq k\leq s   \right\}.$
 By condition (i) and the inequality (\ref{esd7}), the second variation (\ref{esd6}) satisfies for all $\phi\in \mathcal{C}^{\infty}_{0}(\Lambda,U)$ the inequality
 \begin{equation}\label{esd8}
\delta^2 \mathcal{F}(u,\phi)\geq \alpha \int_{\Lambda} \left\|\phi^{(s)}(x) \right\|^2 dx.
\end{equation}
 Consequently, by the second part of \textrm{Theorem} \ref{thm}, $u$ is a weak  minimum for $\mathcal{F}.$
 \end{proof}
\begin{corollary}
Suppose that $L\in \mathcal{C}^2(\Lambda\times \Omega,\mathbb{R}),$  $u$ is a weak  extremum for $\mathcal{F}.$
 If
 \begin{itemize}
 \item[(a)] for all $k\neq k',$ the bilinear forms defined on $\mathbb{R}^{p_k}\times \mathbb{R}^{p_{k'}}$ by the matrices
 $A^{jj'}_{kk'}\left(x,u^{(s)}(x)\right)$ are positive for all $x\in \Lambda,$ and
  \item[(b)] the square matrices $A^{jj'}_{kk}\left(x,u^{(s)}(x)\right)$  are positive definite for all $x\in \Lambda,$
 \end{itemize}
  then $u$ is a strict weak local minimum for $\mathcal{F}.$
\end{corollary}
\begin{proof}
It suffices to show that condition (i) in \textrm{Theorem} \ref{thmmm} is satisfied.
We have
\begin{equation}
J_2=\int_{\Lambda} \sum_{j=1}^{m} \sum_{{}^{j'=1}_{j'\neq j}}^{m}    \sum_{k=0}^{s} \sum_{h,h'=1}^{p_k}  \phi^j_{(k)}[h] \,A^{jj'}_{kk}[h,h']\, {}^{T}\phi^{j'}_{(k)}[h']  \,dx\geq 0
\end{equation}
since by condition (b) the integrand is always positive;
\begin{equation}
I_2=\int_{\Lambda} \sum_{j,j'=1}^{m} \sum_{k=0}^{s} \sum_{{}^{k'=0}_{k'\neq k}}^{s}\sum_{h=1}^{p_k} \sum_{h'=1}^{p_{k'}} \phi^j_{(k)}[h] \,A^{jj'}_{kk'}[h,h'] \,{}^{T}\phi^{j'}_{(k')}[h'] \,dx,
\end{equation}
since by condition (a) the integrand is always positive.
Therefore $J_2+I_2\geq 0$.
\end{proof}

%%%%%%%%%%%%%%%%%%%%%%%%%%%%%%%%%%%%%%%%%%%%%%%%%%%%%%%%%%%%%%%%%%%%%%%%%%%%%%%%%%%%%%%%%%%%%%%%%%%%%%%%%%%%%%%%%%%%%%%

\section{Applications}
To conclude this work, let us analyze some applications.
\begin{example}
Consider the problem of finding extremum point $u=u(x)$ with $x\in [a,b],$ of the functional $\mathcal{F}$ defined by
$$  \mathcal{F}(u)= \int_{a}^{b}\sqrt{1+u_x(x)^2}d x. $$
The Lagrangian of this functional is
$$ L\left(x,u^{(1)}\right) =\sqrt{1+u_x^2}. $$
The extremum must satisfy the Euler-Lagrange equation
$$  \frac{\partial L}{\partial u}-\frac{d}{d x} \left(\frac{\partial L}{\partial u_x} \right)=0  $$
which gives
$$  \frac{u_{x,x}}{\left( 1+u_x^2\right)^{\frac{3}{2}}} =0 .$$
The general solution of this equation is $u(x)=c_1 x+c_2 ,$ where $c_1$ and $c_2$ are constants determined
by the given  end point constraints. \\
Determine the matrix $A$ associated to the second variation of this problem.
$$  A=\left[\begin{array}{cc} A_{00}&  A_{01}\\
A_{10}&A_{11} \end{array}\right],$$
where
$$ A_{00}=\frac{\partial^2 L}{\partial u \partial u} =0,\quad A_{01}=\frac{\partial^2 L}{\partial u \partial u_x} =0,\quad A_{10}=\frac{\partial^2 L}{\partial u_x \partial u} =0,$$
$$A_{11}=\frac{\partial^2 L}{\partial u_x \partial u_x} =\frac{1}{\left( 1+u_x^2\right)^{\frac{3}{2}}}.   $$
Thus,
$$  A=\left[\begin{array}{cc} 0&  0\\
0&\frac{1}{\left( 1+u_x^2\right)^{\frac{3}{2}}} \end{array}\right].$$
It is clear that the matrix $A$ is positive semi-definite. Therefore, the found function $u,$ solution to the
Euler-Lagrange equation, is a minimum point to the functional $\mathcal{F}.$\\
Note here that the Legendre necessary conditions are well satisfied. Indeed, $A_{00}\geq0$ and $A_{11}\geq 0.$
\end{example}
\begin{example}
Consider the problem of finding extremum point $u=u(x)$ with $x\in [a,b],$ of the functional $\mathcal{F}$ defined by
$$  \mathcal{F}(u)= \int_{a}^{b}u(x)\sqrt{1+u_x(x)^2}d x. $$
The Lagrangian of this functional is
$$ L\left(x,u^{(1)}\right) =u\sqrt{1+u_x^2}. $$
The extremum must satisfy the Euler-Lagrange equation
$$  \frac{\partial L}{\partial u}-\frac{d}{d x} \left(\frac{\partial L}{\partial u_x} \right)=0  $$
which gives
$$  \frac{1+u_x^2-u u_{x,x}}{\left( 1+u_x^2\right)^{\frac{3}{2}}} =0 .$$
The general solution of this equation is $u(x)=c_1 \cosh\left(\frac{ x +c_2}{c_1}\right) ,$ where $c_1$ and $c_2$ are constants determined
by the given  end conditions. \\
Determine the matrix $A$ associated with the second variation of this problem.
$$  A=\left[\begin{array}{cc} A_{00}&  A_{01}\\
A_{10}&A_{11} \end{array}\right],$$
where
$$ A_{00}=\frac{\partial^2 L}{\partial u \partial u} =0,\quad A_{01}=\frac{\partial^2 L}{\partial u \partial u_x} =\frac{u_x}{\sqrt{1+u_x^2}},\quad A_{10}=\frac{\partial^2 L}{\partial u_x \partial u} =\frac{u_x}{\sqrt{1+u_x^2}},$$
$$A_{11}=\frac{\partial^2 L}{\partial u_x \partial u_x} =\frac{u}{\left( 1+u_x^2\right)^{\frac{3}{2}}}.   $$
Thus,
$$  A=\left[\begin{array}{cc} 0&  \frac{u_x}{\sqrt{1+u_x^2}}\\
\frac{u_x}{\sqrt{1+u_x^2}}&\frac{u}{\left( 1+u_x^2\right)^{\frac{3}{2}}} \end{array}\right].$$
It is clear that the matrix $A$ is neither positive semi-definite nor negative semi-definite (i.e., $-A$ is not positive semi-definite).
Hence, the found function $u,$ solution to the
Euler-Lagrange equation, is neither   a minimum point nor a maximum point to the functional $\mathcal{F}.$
Therefore, we can conclude that this function $u$ is an instable equilibrium point.
\end{example}
\begin{example}
Let $\Lambda$ be a connected subset of $\mathbb{R}^2.$ Consider the problem of finding
the function $\left(u^1,u^2\right),$ where $u^1=u^1\left(x^1,x^2\right),u^2=u^2\left(x^1,x^2\right)$
with $\left(x^1,x^2\right)\in \Lambda,$ which is an extremum of the functional $\mathcal{F}$ defined by
$$ \mathcal{F}\left(u^1,u^2\right)=\int_{\Lambda}  L\left(x^1,x^2,u^{1(1)}\left(x^1,x^2\right),u^{2(1)}\left(x^1,x^2\right) \right) dx^1 dx^2,$$
where the Lagrangian $L$ is
$$  L=\left(u^1 \right)^2+\left( u^2\right)^2 +\left(u^1_{x^1} \right)^2+\left( u^1_{x^2}\right)^2+\left( u^2_{x^1}\right)^2+
\left(u^2_{x^2} \right)^2+\frac{1}{2}\left(u^1u^2-u^1_{x^1}u^1_{x^2}-u^2_{x^1}u^2_{x^2}\right). $$
Extremum of the functional $\mathcal{F}$ must satisfy the Euler-Lagrange equations
\begin{eqnarray}
 \frac{\partial L}{\partial u^1}-\frac{\partial}{\partial x^1}\left( \frac{\partial L}{\partial u^1_{x^1}} \right)
 -\frac{\partial}{\partial x^2}\left( \frac{\partial L}{\partial u^1_{x^2}} \right)=0& &\nonumber\\
 \frac{\partial L}{\partial u^2}-\frac{\partial}{\partial x^1}\left( \frac{\partial L}{\partial u^2_{x^1}} \right)
 -\frac{\partial}{\partial x^2}\left( \frac{\partial L}{\partial u^2_{x^2}} \right)=0& &\nonumber
\end{eqnarray}
which give the system
\begin{eqnarray}
 2u^1+\frac{1}{2}u^2-2u^1_{2x^1}+u^1_{x^1x^2}-2u^1_{2x^2}=0& &\nonumber\\
 2u^2+\frac{1}{2}u^1-2u^2_{2x^1}+u^2_{x^1x^2}-2u^2_{2x^2}=0.& &\nonumber
\end{eqnarray}
The general solution to this system is
\begin{eqnarray}
 u^1\left(x^1,x^2\right)& =&   c_5e^{-\frac{\sqrt{5}}{2}x^1}+c_6e^{-\frac{\sqrt{5}}{2}x^2}   +  c_7e^{\frac{\sqrt{5}}{2}x^1}+c_8e^{\frac{\sqrt{5}}{2}x^2}   \nonumber\\
 &-&\left(c_1e^{-\frac{\sqrt{3}}{2}x^1}+c_2e^{-\frac{\sqrt{3}}{2}x^2}+c_3e^{\frac{\sqrt{3}}{2}x^1}+c_4e^{\frac{\sqrt{3}}{2}x^2} \right) \nonumber\\
 u^2\left(x^1,x^2\right)& =& c_1e^{-\frac{\sqrt{3}}{2}x^1}+c_2e^{-\frac{\sqrt{3}}{2}x^2}+c_3e^{\frac{\sqrt{3}}{2}x^1}+c_4e^{\frac{\sqrt{3}}{2}x^2}\nonumber\\
 &+&  c_5e^{-\frac{\sqrt{5}}{2}x^1}+c_6e^{-\frac{\sqrt{5}}{2}x^2}  + c_7e^{\frac{\sqrt{5}}{2}x^1}+c_8e^{\frac{\sqrt{5}}{2}x^2}   ,  \nonumber
\end{eqnarray}
where the constants $c_i$ are determined by the given boundary conditions.\\
Determine the matrix $A$ associated with the second variation of the functional $\mathcal{F}.$
$$  A=\left[\begin{array}{cc} A^{11}&  A^{12}\\
A^{21}&A^{22} \end{array}\right]$$
with $A^{jj'}$  defined by
$$  A^{jj'}=\left[\begin{array}{cc} A^{jj'}_{00}&  A^{jj'}_{01}\\
A^{jj'}_{10}&A^{jj'}_{11} \end{array}\right].$$
We have:
$$A^{11}_{00}=\frac{\partial^2 L}{\partial u^1 \partial u^1}=2,\quad A^{12}_{00}=\frac{\partial^2 L}{\partial u^1 \partial u^2}=\frac{1}{2},$$
$$A^{21}_{00}=\frac{\partial^2 L}{\partial u^2 \partial u^1}=\frac{1}{2},\quad A^{22}_{00}=\frac{\partial^2 L}{\partial u^2 \partial u^2}=2,     $$
$$A^{11}_{10}=\left[\begin{array}{c} \frac{\partial^2 L}{\partial u^1_{x^1} \partial u^1} \\ \frac{\partial^2 L}{\partial u^1_{x^2} \partial u^1}  \end{array}\right]  = \left[\begin{array}{c} 0 \\ 0  \end{array}\right], \quad
A^{12}_{10}=\left[\begin{array}{c} \frac{\partial^2 L}{\partial u^1_{x^1} \partial u^2} \\ \frac{\partial^2 L}{\partial u^1_{x^2} \partial u^2}  \end{array}\right]  = \left[\begin{array}{c} 0 \\ 0  \end{array}\right],$$
$$A^{21}_{10}=\left[\begin{array}{c} \frac{\partial^2 L}{\partial u^2_{x^1} \partial u^1} \\ \frac{\partial^2 L}{\partial u^2_{x^2} \partial u^1}  \end{array}\right]  = \left[\begin{array}{c} 0 \\ 0  \end{array}\right], \quad
A^{22}_{10}=\left[\begin{array}{c} \frac{\partial^2 L}{\partial u^2_{x^1} \partial u^2} \\ \frac{\partial^2 L}{\partial u^2_{x^2} \partial u^2}  \end{array}\right]  = \left[\begin{array}{c} 0 \\ 0  \end{array}\right],$$
$$A^{11}_{01}=\left[\begin{array}{cc} \frac{\partial^2 L}{\partial u^1 \partial u^1_{x^1}} & \frac{\partial^2 L}{\partial u^1 \partial u^1_{x^2}}  \end{array}\right]  = \left[\begin{array}{cc} 0 & 0  \end{array}\right], \quad
A^{12}_{01}=\left[\begin{array}{cc} \frac{\partial^2 L}{\partial u^1 \partial u^2_{x^1}} & \frac{\partial^2 L}{\partial u^1 \partial u^2_{x^2}}  \end{array}\right]  = \left[\begin{array}{cc} 0 & 0  \end{array}\right], $$
$$A^{21}_{01}=\left[\begin{array}{cc} \frac{\partial^2 L}{\partial u^2 \partial u^1_{x^1}} & \frac{\partial^2 L}{\partial u^2 \partial u^1_{x^2}}  \end{array}\right]  = \left[\begin{array}{cc} 0 & 0  \end{array}\right], \quad
A^{22}_{01}=\left[\begin{array}{cc} \frac{\partial^2 L}{\partial u^2 \partial u^2_{x^1}} & \frac{\partial^2 L}{\partial u^2 \partial u^2_{x^2}}  \end{array}\right]  = \left[\begin{array}{cc} 0 & 0  \end{array}\right], $$
$$A^{12}_{11}=\left[\begin{array}{cc} \frac{\partial^2 L}{\partial u^1_{x^1} \partial u^2_{x^1}} & \frac{\partial^2 L}{\partial u^1_{x^1} \partial u^2_{x^2}} \\
\frac{\partial^2 L}{\partial u^1_{x^2} \partial u^2_{x^1}} & \frac{\partial^2 L}{\partial u^1_{x^2} \partial u^2_{x^2}} \end{array}\right]  = \left[\begin{array}{cc} 0 & 0\\
0&0  \end{array}\right],$$
$$A^{21}_{11}=\left[\begin{array}{cc} \frac{\partial^2 L}{\partial u^2_{x^1} \partial u^1_{x^1}} & \frac{\partial^2 L}{\partial u^2_{x^1} \partial u^1_{x^2}} \\
\frac{\partial^2 L}{\partial u^2_{x^2} \partial u^1_{x^1}} & \frac{\partial^2 L}{\partial u^2_{x^2} \partial u^1_{x^2}} \end{array}\right]  = \left[\begin{array}{cc}0 & 0\\
0&0   \end{array}\right],  $$
$$A^{11}_{11}=\left[\begin{array}{cc} \frac{\partial^2 L}{\partial u^1_{x^1} \partial u^1_{x^1}} & \frac{\partial^2 L}{\partial u^1_{x^1} \partial u^1_{x^2}} \\
\frac{\partial^2 L}{\partial u^1_{x^2} \partial u^1_{x^1}} & \frac{\partial^2 L}{\partial u^1_{x^2} \partial u^1_{x^2}} \end{array}\right]  = \left[\begin{array}{cc} 2 & -\frac{1}{2}\\
-\frac{1}{2}&2  \end{array}\right], $$
$$A^{22}_{11}=\left[\begin{array}{cc} \frac{\partial^2 L}{\partial u^2_{x^1} \partial u^2_{x^1}} & \frac{\partial^2 L}{\partial u^2_{x^1} \partial u^2_{x^2}} \\
\frac{\partial^2 L}{\partial u^2_{x^2} \partial u^2_{x^1}} & \frac{\partial^2 L}{\partial u^2_{x^2} \partial u^2_{x^2}} \end{array}\right]  = \left[\begin{array}{cc} 2 & -\frac{1}{2}\\
-\frac{1}{2}&2 \end{array}\right],$$
which give
$$  A^{12}= \left[\begin{array}{ccc}\frac{1}{2}&0&0\\
  0&0&0\\
  0&0&0\end{array}\right]   =A^{21} , \quad
  A^{11}=\left[\begin{array}{ccc}2&0&0\\
  0&2&-\frac{1}{2}\\
  0&-\frac{1}{2}&2  \end{array}\right]    =A^{22} $$
and hence
$$ A= \left[\begin{array}{cccccc}2&0&0&\frac{1}{2}&0&0\\
0&2&-\frac{1}{2}&2&0&0\\
0&-\frac{1}{2}&2&0&0&0\\
\frac{1}{2}&0&0&2&0&0\\
0&0&0&0&2&-\frac{1}{2}\\
0&0&0&0&-\frac{1}{2}&2  \end{array}\right]   .$$
It is easy to see that the matrices $A^{jj'}_{kk},$ $k=0,1,$ are all semi-positive definite. This implies that
the Legendre necessary conditions for minimum point are satisfied. Furthermore, the matrix $A$ is positive definite.
Thus, we can well conclude that  the found solution $\left(u^1,u^2 \right)$ to the Euler-Lagrange equations is
effectively a minimum point for the functional $\mathcal{F}.$
\end{example}
\begin{example}
Let $\Lambda$ be a connected subset of $\mathbb{R}^2.$ Consider the problem of finding a function
$u=u\left(x^1,x^2\right)$ with $\left(x^1,x^2\right)\in \Lambda,$ which is an extremum  of the functional $\mathcal{F}$
defined by
$$  \mathcal{F}(u)=\int_{\Lambda}  L\left( x^1,x^2,u^{(2)}\left(x^1,x^2\right) \right)  dx^1dx^2   $$
whose Lagrangian $L$ is
\begin{eqnarray}
L&=&u^2+u^2_{x^1}+u^2_{x^2}+u_{2x^1}^2+u_{x^1x^2}^2+u_{2x^2}^2\nonumber\\
&-&\frac{1}{2}\left(u_{x^1}u_{x^2}+u_{2x^1}u_{x^1x^2}+u_{2x^1}u_{2x^2}+u_{x^1x^2}u_{2x^2}  \right)  .
\end{eqnarray}
The extremum must satisfy the Euler-Lagrange equation
\begin{eqnarray}
0&=&\frac{\partial L}{\partial u}-\frac{\partial}{\partial x^1}\left(\frac{\partial L}{\partial u_{x^1}} \right)
-\frac{\partial}{\partial x^2}\left(\frac{\partial L}{\partial u_{x^2}} \right)
+\frac{\partial}{\partial x^1}\frac{\partial}{\partial x^1}\left(\frac{\partial L}{\partial u_{2x^1}} \right)\nonumber\\
&+&\frac{\partial}{\partial x^1}\frac{\partial}{\partial x^2}\left(\frac{\partial L}{\partial u_{x^1x^2}} \right)
+\frac{\partial}{\partial x^2}\frac{\partial}{\partial x^2}\left(\frac{\partial L}{\partial u_{2x^2}} \right)
\end{eqnarray}
which gives the equation
$$  2u-2u_{2x^1}+u_{x^1x^2}-2u_{2x^2}+2u_{4x^1}-u_{3x^1x^2}+u_{2x^12x^2}-u_{x^13x^2}+2u_{4x^2}=0.    $$
The general solution to this equation is
\begin{eqnarray}
u\left(x^1,x^2\right)&=&e^{-\frac{\sqrt{3}}{2}x^1}\left[ c_1\cos\left(\frac{1}{2}x^1\right)+c_2\sin\left(\frac{1}{2}x^1\right) \right]\nonumber\\
&+&e^{\frac{\sqrt{3}}{2}x^1}\left[ c_3\cos\left(\frac{1}{2}x^1\right)+c_4\sin\left(\frac{1}{2}x^1\right) \right]\nonumber\\
&+&e^{-\frac{\sqrt{3}}{2}x^2}\left[ c_5\cos\left(\frac{1}{2}x^2\right)+c_6\sin\left(\frac{1}{2}x^2\right) \right]\nonumber\\
&+&e^{\frac{\sqrt{3}}{2}x^2}\left[ c_7\cos\left(\frac{1}{2}x^2\right)+c_8\sin\left(\frac{1}{2}x^2\right) \right],\nonumber
\end{eqnarray}
where the constants $c_i$ are determined by the boundary conditions.\\
Determine the matrix $A$ associated with the second variation of $\mathcal{F}:$
$$ A=\left[ \begin{array}{ccc}A_{00}&A_{01}&A_{02}\\
A_{10}&A_{11}&A_{12}\\
A_{20}&A_{21}&A_{22}     \end{array}\right], $$
where
$$   A_{00}=\frac{\partial^2 L}{\partial u \partial u}=2,\quad A_{02}=\left[\begin{array}{ccc}\frac{\partial^2 L}{\partial u \partial u_{2x^1}}&\frac{\partial^2 L}{\partial u \partial u_{x^1x^2}}&\frac{\partial^2 L}{\partial u \partial u_{2x^2}}  \end{array} \right]= \left[\begin{array}{ccc}  0&0&0\end{array} \right] ,    $$
$$ A_{10}=\left[\begin{array}{c}\frac{\partial^2 L}{\partial u_{x^1} \partial u}\\ \frac{\partial^2 L}{\partial u_{x^2} \partial u}  \end{array} \right]= \left[\begin{array}{c}  0\\0\end{array} \right] , \quad A_{01}=\left[\begin{array}{cc}\frac{\partial^2 L}{\partial u \partial u_{x^1}}&\frac{\partial^2 L}{\partial u \partial u_{x^2}}  \end{array} \right]= \left[\begin{array}{cc}  0&0\end{array} \right] ,     $$
$$ A_{11}=\left[\begin{array}{cc} \frac{\partial^2 L}{\partial u_{x^1} \partial u_{x^1}}&\frac{\partial^2 L}{\partial u_{x^1} \partial u_{x^2}}\\
\frac{\partial^2 L}{\partial u_{x^2} \partial u_{x^1}} &\frac{\partial^2 L}{\partial u_{x^2} \partial u_{x^2}}     \end{array} \right]=\left[\begin{array}{cc}  2&-\frac{1}{2}\\
-\frac{1}{2}&2     \end{array} \right],\quad
A_{20}= \left[\begin{array}{c} \frac{\partial^2 L}{\partial u_{2x^1} \partial u} \\
  \frac{\partial^2 L}{\partial u_{x^1x^2} \partial u}\\
  \frac{\partial^2 L}{\partial u_{2x^2} \partial u}\end{array} \right]  =\left[\begin{array}{c} 0\\0\\0   \end{array} \right]   ,$$
$$A_{12}=  \left[\begin{array}{ccc} \frac{\partial^2 L}{\partial u_{x^1} \partial u_{2x^1}}&\frac{\partial^2 L}{\partial u_{x^1} \partial u_{x^1x^2}}&\frac{\partial^2 L}{\partial u_{x^1} \partial u_{2x^2}} \\
  \frac{\partial^2 L}{\partial u_{x^2} \partial u_{2x^1}}&\frac{\partial^2 L}{\partial u_{x^2} \partial u_{x^1x^2}}&\frac{\partial^2 L}{\partial u_{x^2} \partial u_{2x^2}} \end{array} \right]=  \left[\begin{array}{ccc}  0&0&0\\0&0&0  \end{array} \right], $$
$$A_{21}=\left[\begin{array}{cc} \frac{\partial^2 L}{\partial u_{2x^1} \partial u_{x^1}}&\frac{\partial^2 L}{\partial u_{2x^1} \partial u_{x^2}}\\
 \frac{\partial^2 L}{\partial u_{x^1x^2} \partial u_{x^1}}&\frac{\partial^2 L}{\partial u_{x^1x^2} \partial u_{x^2}}\\
   \frac{\partial^2 L}{\partial u_{2x^2} \partial u_{x^1}}&\frac{\partial^2 L}{\partial u_{2x^2} \partial u_{x^2}}\end{array} \right]=\left[\begin{array}{cc}  0&0\\0&0\\0&0  \end{array} \right]   ,$$
$$       A_{22}=\left[\begin{array}{ccc}\frac{\partial^2 L}{\partial u_{2x^1} \partial u_{2x^1}}&\frac{\partial^2 L}{\partial u_{2x^1} \partial u_{x^1x^2}}&\frac{\partial^2 L}{\partial u_{2x^1} \partial u_{2x^2}} \\
 \frac{\partial^2 L}{\partial u_{x^1x^2} \partial u_{2x^1}}&\frac{\partial^2 L}{\partial u_{x^1x^2} \partial u_{x^1x^2}}&\frac{\partial^2 L}{\partial u_{x^1x^2} \partial u_{2x^2}}\\
 \frac{\partial^2 L}{\partial u_{2x^2} \partial u_{2x^1}}&\frac{\partial^2 L}{\partial u_{2x^2} \partial u_{x^1x^2}}&\frac{\partial^2 L}{\partial u_{2x^2} \partial u_{2x^2}}\end{array} \right]= \left[\begin{array}{ccc}  2&-\frac{1}{2}&-\frac{1}{2}\\
-\frac{1}{2}&2&-\frac{1}{2}\\
-\frac{1}{2}&-\frac{1}{2}&2\end{array} \right].   $$
Thus,
$$ A= \left[\begin{array}{cccccc} 2&0&0&0&0&0\\
0&2&-\frac{1}{2}&0&0&0\\
0&-\frac{1}{2}&2&0&0&0\\
0&0&0&2&-\frac{1}{2}&-\frac{1}{2}\\
0&0&0&-\frac{1}{2}&2&-\frac{1}{2}\\
0&0&0&-\frac{1}{2}&-\frac{1}{2}&2  \end{array} \right].$$
It is easy to see that the matrices $A_{kk},$ $k=0,1,2$ are all semi-positive definite. This implies that
the Legendre necessary conditions for minimum point are satisfied. Furthermore, the matrix $A$ is positive definite.
Thus, we can well conclude that  the found solution $u$ to the Euler-Lagrange equations is
effectively a minimum point for the functional $\mathcal{F}.$
\end{example}

\subsection*{Acknowledgments}
This work is partially supported by the ICTP through the
OEA-ICMPA-Prj-15. The ICMPA is in partnership with
the Daniel Iagolnitzer Foundation (DIF), France.

\label{lastpage-01}
\end{document}